\documentclass[11pt]{article}

\title{Monotone Submodular Matroid}

\usepackage{etoolbox}
\usepackage{graphicx}
\usepackage{latexsym}
\usepackage{amssymb}
\usepackage{amsmath}
\usepackage{amsthm}
\usepackage{thmtools}
\usepackage[margin=1in]{geometry}
\usepackage[ruled,vlined,linesnumbered]{algorithm2e}
\usepackage{forloop}
\usepackage{mleftright}\mleftright
\usepackage{xcolor}
\usepackage{nicefrac}
\usepackage{paralist}
\usepackage{hyperref}

\makeatletter
\patchcmd{\@algocf@start}
  {-1.5em}
  {0pt}
  {}{}
\makeatother

\newcommand{\newnl}{\let\nl\oldnl}
\newcommand{\nonl}{\let\oldnl\nl \let\nl\newnl}

\definecolor{darkgreen}{rgb}{0,0.5,0}
\hypersetup{
    unicode=false,            
    colorlinks=true,          
    linkcolor=blue!70!black,  
    citecolor=darkgreen,      
    filecolor=magenta,        
    urlcolor=blue!70!black    
}

\newtheorem{theorem}{Theorem}[section]
\newtheorem{lemma}[theorem]{Lemma}

\newtheorem{corollary}[theorem]{Corollary}
\newtheorem{definition}{Definition}[section]
\newtheorem{proposition}[theorem]{Proposition}

\newtheorem{observation}[theorem]{Observation}

\newtheorem{remark}[theorem]{Remark}

\newcommand{\defcal}[1]{\expandafter\newcommand\csname c#1\endcsname{{\mathcal{#1}}}}
\newcommand{\defbb}[1]{\expandafter\newcommand\csname b#1\endcsname{{\mathbb{#1}}}}
\newcounter{calBbCounter}
\forLoop{1}{26}{calBbCounter}{
    \edef\letter{\Alph{calBbCounter}}
    \expandafter\defcal\letter
		\expandafter\defbb\letter
}

\newcommand{\eps}{\varepsilon}
\newcommand{\nnR}{{\bR_{\geq 0}}}
\newcommand{\cupdot}{\mathbin{\mathaccent\cdot\cup}}
\newcommand{\email}[1]{{\href{mailto:#1}{#1}}}
\newcommand{\ignore}[1]{}
\DeclareMathOperator{\Poly}{Poly}

\title{Deterministic Algorithm and Faster Algorithm for Submodular Maximization subject to a Matroid Constraint}

\author{Niv Buchbinder\thanks{Department of Statistics and Operations Research, Tel Aviv University. E-mail: \email{niv.buchbinder@gmail.com}} \and
				Moran Feldman\thanks{Department of Computer Science, University of Haifa. E-mail: \email{moranfe@cs.haifa.ac.il}}}

\begin{document}

\maketitle
\thispagestyle{empty}
\pagenumbering{Alph}
\begin{abstract}
We study the problem of maximizing a monotone submodular function subject to a matroid constraint, and
present for it a {\bf deterministic} non-oblivious local search algorithm that has an approximation guarantee of $1 - \nicefrac{1}{e} - \eps$ (for any $\eps > 0$) and query complexity of $\tilde{O}_\eps(nr)$, where $n$ is the size of the ground set and $r$ is the rank of the matroid. Our algorithm vastly improves over the previous state-of-the-art $0.5008$-approximation deterministic algorithm, and in fact, shows that there is no separation between the approximation guarantees that can be obtained by deterministic and randomized algorithms for the problem considered. The query complexity of our algorithm can be improved to $\tilde{O}_\eps(n + r\sqrt{n})$ using randomization, which is nearly-linear for $r = O(\sqrt{n})$, and is always at least as good as the previous state-of-the-art algorithms. 

\medskip


\end{abstract}
\newpage
\pagenumbering{arabic}

\section{Introduction}

The problem of maximizing a non-negative monotone submodular function subject to a matroid constraint is one of earliest and most studied problems in submodular maximization. 
The natural greedy algorithm obtains the optimal approximation ratio of $1 - \nicefrac{1}{e} \approx 0.632$ for this problem when the matroid constraint is restricted to be a cardinality constraint~\cite{nemhauser1978best,nemhauser1978analysis}. However, the situation turned out to be much more involved for general matroid constraints (or even partition matroid constraints). In $1978$, Fisher et al. showed that the natural greedy algorithm guarantees a sub-optimal approximation ratio of $\nicefrac{1}{2}$ for such constraints. This was improved (roughly $40$ years later) by C{\u{a}}linescu et al.~\cite{calinescu2011maximzing}, who described a Continuous Greedy algorithm obtaining the optimal approximation ratio of $1 - \nicefrac{1}{e}$ for general matroid constraints. The Continuous Greedy algorithm was a celebrated major improvement over the state-of-the-art, but it has a significant drawback; namely, it is based on a continuous extension of set functions termed the \emph{multilinear extension}. The only known way to evaluate the multilinear extension is by sampling sets from an appropriate distribution, which makes the continuous greedy algorithm (and related algorithms suggested since the work of C{\u{a}}linescu et al.~\cite{calinescu2011maximzing}) both randomized and quite slow.

Badanidiyuru \& Vondr{\'{a}}k~\cite{babanidiyuru2014fast} were the first to study ways to improve over the time complexity of C{\u{a}}linescu et al.~\cite{calinescu2011maximzing}. To understand their work, and works that followed it, we first need to mention that it is standard in the literature to assume that access to the submodular objective function and the matroid constraint is done via {\em value} and {\em independence oracles} (see Section~\ref{sec:preliminaries} for details). The number of queries to these oracles used by the algorithm is often used as a proxy for its time complexity. Determining the exact time complexity is avoided since it requires taking care of technical details, such as the data structures used to maintain sets, and doing so makes more sense in the context of particular classes of matroid constraints (see, for example,~\cite{ene2019towards,henzinger2023faster}). 

In their above-mentioned work, Badanidiyuru \& Vondr{\'{a}}k~\cite{babanidiyuru2014fast} gave a variant of the Continuous Greedy algorithm that obtains $1 - \nicefrac{1}{e} - \eps$ approximation using $\tilde{O}_\eps(nr)$ queries,\footnote{The {\raise.17ex\hbox{$\scriptstyle\sim$}} sign in $\tilde{O}_\eps$ suppresses poly-logarithmic factors, and the $\eps$ subscript suppresses factors depending only on $\eps$.} where $n$ is the size of the ground set, and $r$ is the rank of the matroid constraint (the maximum size of a feasible solution). Later, Buchbinder et al.~\cite{buchbinder2017comparing} showed that, by combining the algorithm of~\cite{babanidiyuru2014fast} with the Residual Random Greedy algorithm of~\cite{buchbinder2014submodular}, it is possible to get the same approximation ratio using $\tilde{O}_\eps(n\sqrt{r} + r^2)$ queries. Very recently, an improved rounding technique has allowed Kobayashi \& Terao~\cite{kobayash2024sunquadratic} to drop the $r^2$ term from the last query complexity, which is an improvement in the regime of $r = \omega(n^{2/3})$. While the query complexity of~\cite{kobayash2024sunquadratic} is the state-of-the-art for general matroid constraints, some works have been able to further improve the query and time complexities to be nearly-linear (i.e., $\tilde{O}_\eps(n)$) in the special cases of partition matroids, graphic matroids~\cite{ene2019towards}, laminar matroids and transversal matroids~\cite{henzinger2023faster}. It was also shown by~\cite{kobayash2024sunquadratic} that a query complexity of $\tilde{O}_\eps(n + r^{3/2})$ suffices when we are given access to the matroid constraint via an oracle known as \emph{rank oracle}, which is a stronger oracle compared to the independence oracle assumed in this paper (and the vast majority of the literature).

All the above results are randomized as they inherit the use of the multilinear extension from the basic Continuous Greedy algorithm. Filmus and Ward~\cite{filmus2014monotone} were the first to try to tackle the inherent drawbacks of the multilinear extension by suggesting a non-oblivious local search algorithm.\footnote{A local search algorithm is \emph{non-oblivious} if the objective function it tries to optimize is an auxiliary function rather than the true objective function of the problem.} Their algorithm is more combinatorial in nature as it does not require a rounding step. However, it is still randomized because the evaluation of its auxiliary objective function is done by sampling from an appropriate distribution (like in the case for the multilinear extension). To the best of our knowledge, the only known deterministic algorithm for the problem we consider whose approximation ratio improves over the $\nicefrac{1}{2}$-approximation guarantee of the natural greedy algorithm is an algorithm due to Buchbinder et al.~\cite{buchbinder2023deterministic} that guarantees $0.5008$-approximation.

\subsection{Our Contribution} \label{ssc:contribution}

We suggest an (arguably) simple deterministic algorithm for maximizing  a non-negative monotone submodular function $f\colon 2^\cN \rightarrow \nnR$ subject to a matroid constraint $\cM = (\cN, \cI)$.\footnote{See Section~\ref{sec:preliminaries} for the formal definitions of non-negative monotone submodular functions and matroid constraints.} Our algorithm is a non-oblivious local search algorithm parametrized by a positive integer value $\ell$ and additional positive values $\alpha_1, \ldots, \alpha_{\ell}$. The algorithm maintains $\ell$ disjoint sets $S_1, S_2, \ldots, S_{\ell}$ whose union is a base of $\cM$ (i.e., an inclusion-wise maximal feasible set), and aims to maximize an auxiliary function $g(S_1, \ldots, S_{\ell})$. To state this function, it is useful to define, for any $J\subseteq [\ell]\triangleq \{1, 2, \dotsc, \ell\}$, the shorthand $S_{J} \triangleq \cupdot_{i\in J}S_j$ (the dot inside the union is only aimed to emphasize that the sets on which the union is done are disjoint). For convenience, we also define $\alpha_0 \triangleq 0$. Then, the auxiliary function $g$ is given by 
\[g(S_1, \ldots, S_{\ell}) =  \sum_{J \subseteq [\ell]} \alpha_{|J|} \cdot f(S_J) \enspace.\]

\begin{algorithm}[h]
\caption{\textsc{Non-Oblivious Local Search}} \label{alg:basic}
    Start from any $\ell$ disjoint sets $S_1, \ldots, S_{\ell}$ such that $S_{[\ell]}$ is a base of the matroid $\cM$.\\
    While some step from the following list strictly increases the value of 
$g(S_1, \ldots, S_{\ell})$, perform such a step.\\\nonl
\vphantom{g}\begin{minipage}[t]{15.5cm}
\begin{compactitem}
    \item Step 1: Move some element $u\in S_k$ to a different set $S_j$. \label{alg:step1}
    \item Step 2: Remove some element $u\in S_k$, and add to any $S_j$ some element $v\in \cN\setminus S_{[\ell]}$ such that $(S_{[\ell]}\setminus \{u\}) \cup \{v\} \in \cI$. \label{alg:step2}
\end{compactitem}
\end{minipage}\\
    \Return $S_{[\ell]}$.
\end{algorithm}

The formal definition of our algorithm appears as Algorithm~\ref{alg:basic}. Its output is the set $S_{[\ell]}$, which is clearly a base of the matroid $\cM$. The approximation guarantee of Algorithm~\ref{alg:basic} is given by the next proposition. When $\ell=1$, Algorithm~\ref{alg:basic} reduces to the standard local search algorithm (applied to $f$), and Proposition~\ref{prop:approximation_basic} recovers in this case the known $\nicefrac{1}{2}$-approximation guarantee of this local search algorithm due to~\cite{fisher1978analysis}. On the other extreme, when $\ell = r$, Algorithm~\ref{alg:basic} reduces to a (non-sampling exponential time) version of the algorithm of Filmus and Ward~\cite{filmus2014monotone}.
\newtoggle{introduction}\toggletrue{introduction}
\begin{restatable}{proposition}{propApproximationBasic} \label{prop:approximation_basic}
\iftoggle{introduction}{Let $OPT\in \cI$ be a set  maximizing $f$. }{}By setting $\alpha_i = \frac{(1 + 1/\ell)^{i - 1}}{\binom{\ell - 1}{i - 1}}$ for every $i \in [\ell]$, it is guaranteed that
    \[f(S_{[\ell]}) \geq \left(1 - (1 + \nicefrac{1}{\ell})^{-\ell}\right)\cdot f(OPT) + (1 + \nicefrac{1}{\ell})^{-\ell} \cdot f(\varnothing)
    \enspace. \]
Thus, for every $\eps > 0$, one can get $f(S_{[\ell]}) \geq\left(1 - \nicefrac{1}{e} - O(\eps)\right) \cdot f(OPT)$ by setting $\ell = 1 + \lceil 1/\eps \rceil$.
\end{restatable}

Unfortunately, the query complexity of Algorithm~\ref{alg:basic} is not guaranteed to be polynomial (for any value of $\ell$), which is often the case with local search algorithms that make changes even when these changes lead only to a small marginal improvement in the objective. To remedy this, we show how to implement a fast variant of Algorithm~\ref{alg:basic}, which proves our main theorem. 

\begin{theorem} \label{thm:main}
There exists a deterministic non-oblivious local search algorithm for maximizing a non-negative monotone submodular function subject to a matroid constraint that has an approximation guarantee of $1 - \nicefrac{1}{e} - O(\eps)$ (for any $\eps > 0$) and query complexity of $\tilde{O}_\eps(nr)$, where $r$ is the rank of the matroid and $n$ is the size of its ground set. The query complexity can be improved to $\tilde{O}_\eps(n + r\sqrt{n})$ using randomization. 
\end{theorem}

Our algorithm vastly improves over the previous state-of-the-art $0.5008$-approximation deterministic algorithm, and in fact, shows that there is no separation between the approximation guarantees that can be obtained by deterministic and randomized algorithms for maximizing monotone submodular functions subject to a general matroid constraint. In terms of the query complexity, our (randomized) algorithm is nearly-linear for $r = \tilde{O}(\sqrt{n})$, and its query complexity is always at least as good as the $\tilde{O}_\eps(n\sqrt{r})$ query complexity of the previous state-of-the-art algorithm for general matroid constraints (due to~\cite{kobayash2024sunquadratic}).

A central component of our algorithm is a procedure that given a submodular function $f\colon 2^\cN \rightarrow \bR$, a matroid $\cM = (\cN, \cI)$, a value $\eps\in (0,1)$, and a set $S_0 \in \cI$ outputs a set $S\in \cI$
satisfying 
\begin{equation}
\sum_{v\in T}f(v \mid S \setminus \{v\})- \sum_{u\in S}f(u \mid S \setminus \{u\}) \leq \eps \cdot [f(OPT) - f(S_0)] \qquad \forall T\in \cI \enspace,\label{ineq-localOPT1}
\end{equation}
where $f(v \mid S) \triangleq f(S \cup \{v\}) - f(S)$ represents the marginal contribution of element $v$ to the set $S$. 
We note that the above set $S$ is not a local maximum of $f$ in the usual sense (see Remark \ref{rem:localOPT}). 
However, the weaker guarantee given by Inequality~\eqref{ineq-localOPT1} suffices for our purposes, and can be obtained faster.
The following proposition states the query complexity required to get such a set $S$. 
We believe that this proposition may be of independent interest.

\begin{restatable}{proposition}{propFast} \iftoggle{introduction}{\label{prop:fast}}{}
    Let $f\colon 2^\cN \rightarrow \bR$ be a submodular function, $M=(\cN,\cI)$ be a matroid of rank $r$ over a ground set $\cN$ of size $n$, \iftoggle{introduction}{and }{}$S_0$ be a set in $\cI$\iftoggle{introduction}{}{, and $OPT$ be a set in $\cI$ maximizing $f$}. Then, for any $\eps\in (0,1)$, 
    \begin{itemize}
        \item there exists a {\bf deterministic} algorithm that makes $\tilde{O}(\eps^{-1}nr)$ value and independence oracle queries and outputs a subset $S\in \cI$ satisfying \iftoggle{introduction}{Inequality~\eqref{ineq-localOPT1}.}{\begin{equation} \sum_{v\in T}f(v \mid S - v)- \sum_{u\in S}f(u \mid S-u) \leq \eps \cdot [f(OPT) - f(S_0)] \qquad \forall T\in \cI \enspace. \forcelabel{ineq-localOPT}\end{equation}} 
    \item there exists a {\bf randomized} algorithm that makes $\tilde{O}(\eps^{-1}(n+r\sqrt{n}))$ value and independence oracle queries and outputs a subset $S\in \cI$ that with probability at least $1-\eps$ satisfies Inequality~\iftoggle{introduction}{\eqref{ineq-localOPT1}}{\eqref{ineq-localOPT}}. Furthermore, if $f$ is guaranteed to be also monotone and $S_0 = \varnothing$, then the algorithm can be modified so that it outputs ``Fail" whenever $S$ does not satisfy Inequality~\iftoggle{introduction}{\eqref{ineq-localOPT1}}{\eqref{ineq-localOPT}} (the algorithm might output ``Fail'' also when  Inequality~\iftoggle{introduction}{\eqref{ineq-localOPT1}}{\eqref{ineq-localOPT}} is satisfied, but it outputs ``Fail'' with probability at most $\eps$ overall).
		\end{itemize}
\end{restatable}
\togglefalse{introduction}


A previous version of this paper~\cite{buchbinder2024deterministic} used a slightly weaker version of Proposition~\ref{prop:fast} in which the function $f$ was restricted to be non-negative and monotone, the term $f(v \mid S \setminus \{v\})$ in Inequality~\eqref{ineq-localOPT1} was replaced with $f(v \mid S)$ and the term $f(S_0)$ did not appear at all in Inequality~\eqref{ineq-localOPT1}. This weaker version was sufficient for getting the main results of this paper, but had two intuitive weaknesses: first, it applied only to a subset of all submodular function, and second, a set $S \in \cI$ obeying the old version of Inequality~\eqref{ineq-localOPT1} with respect to a linear function $f$ was not necessarily close to maximizing $f$ among all the sets of $\cI$. These weaknesses are important when one would like to maximize sums of linear and monotone submodular functions. Maximization of such sums can be used to maximize monotone submodular functions with bounded curvature~\cite{sviridenko2017optimal} and to capture soft constraints and regularizers in machine learning applications~\cite{harshaw2019submodular,kazemi2021regularized,nikolakaki2021efficient}. See Appendix~\ref{app:submodular_linear_sums} for more detail.

Following this work, Buchbinder \& Feldman~\cite{buchbinder2024extending} developed a greedy-based \emph{deterministic} algorithm for maximizing general submodular functions subject to a matroid constraint. For the special case of monotone submodular functions, the algorithm of~\cite{buchbinder2024extending} recovers the optimal approximation guarantee of Theorem~\ref{thm:main}. Like our algorithm, the algorithm of~\cite{buchbinder2024extending} also has both deterministic and randomized versions. The deterministic version has a much larger query complexity compared to the query complexity stated in Theorem~\ref{thm:main} for the deterministic case. However, the randomized version of the algorithm of~\cite{buchbinder2024extending} has a query complexity of $\tilde{O}_\eps(n + r^{3/2})$, which is always at least as good as the query complexity our algorithm.

\subsection{Additional Related Work}

As mentioned above, the standard greedy algorithm obtains the optimal approximation ratio of $1 - \nicefrac{1}{e}$ for the problem of maximizing a monotone submodular function subject to a cardinality constraint. It was long known that this algorithm can be sped in practice using lazy evaluations, and Badanidiyuru and Vondr{\'{a}}k~\cite{babanidiyuru2014fast} were able to use this idea to reduce the time complexity of the greedy algorithm to be nearly linear. Later, Mirzasoleiman et al.~\cite{mirzasoleiman2015lazier} and Buchbinder et al.~\cite{buchbinder2017comparing} were able to reduce the time complexity to be cleanly linear at the cost of using randomization and deteriorating the approximation ratio by $\eps$. Obtaining the same time complexity and approximation ratio using a deterministic algorithm was an open question that was recently settled by Kuhnle~\cite{kuhnle2021}, Li et al.~\cite{li2022submodular} and Huang \& Kakimura~\cite{huang2022multipass}.

We have discussed above works that aim to speed up the maximization of a monotone submodular function subject to a matroid constraint in the standard offline computational model. Other works have tried to get such a speed up by considering different computational models. Mirzasoleiman et al.~\cite{mirzasoleiman2016distributed} initiated the study of submodular maximization under a distributed (or Map-Reduce like) computational model. Barbosa et al.~\cite{pontebarbosa2015power,pontebarbosa2016new} then showed that the guarantees of both the natural greedy algorithm and Continuous Greedy can be obtained in this computational model (see also~\cite{gopal2024greedyml} for a variant of the algorithm of~\cite{pontebarbosa2016new} that has a reduced memory requirement). Balkanski \& Singer~\cite{balkanski2018adaptive1} initiated the study of algorithms that are parallelizable in a different sense. Specifically, they looked for submodular maximization algorithms whose objective function evaluations can be batched into a small number of sets such that the evaluations that should belong to each set can be determined based only on the results of evaluations from previous sets, which means that the evaluations of each set can be computed in parallel. Algorithms that have this property are called \emph{low adaptivity algorithms}, and such algorithms guaranteeing $(1 - \nicefrac{1}{e} - \eps)$-approximation for maximizing a monotone submodular function subject to a matroid constraint were obtained by Balkanski et al.~\cite{balkanski2022adaptive2}, Chekuri \& Quanrud~\cite{chekuri2019parallelizing}, and Ene et al.~\cite{ene2019submodular}. Interestingly, Li et al.~\cite{li2020polynomial} proved that no \emph{low adaptivity algorithm} can get a clean $(1 - \nicefrac{1}{e})$-approximation for the above problem, and getting such an approximation guarantee requires a polynomial number of sets of evaluations.

Online and streaming versions of submodular maximization problems have been studied since the works of Buchbinder et al.~\cite{buchbinder2019online} and Chakrabarti \& Kale~\cite{chakrabarti2015submodular}, respectively. The main motivations for these studies were to get algorithms that either can make decisions under uncertainty, or can process the input sequentially using a low space complexity. However, online and streaming algorithms also tend to be very fast. Specifically, the streaming algorithm of Feldman et al.~\cite{feldman2022streaming} guarantees $0.3178$-approximation for maximizing a monotone submodular function subject to a matroid constraint using a nearly-linear query complexity. The same approximation ratio was obtained by an online algorithm in the special cases of partition matroid and cardinality constraints~\cite{chan2018online}. See also~\cite{chekuri2015streaming,harshaw20K22power} for streaming algorithms that can handle intersections of multiple matroid constraints.

We conclude this section by mentioning the rich literature on maximization of submodular functions that are not guaranteed to be monotone. Following a long line of work~\cite{buchbinder2019constrained,chekuri2014submodular,ene2016constrained,feldman2011unified}, the state-of-the-art algorithm for the problem of maximizing a (not necessarily monotone) submodular function subject to a matroid constraint guarantees $0.401$-approximation~\cite{buchbinder2023constrained}. Oveis Gharan and Vondr{\'{a}}k~\cite{gharan2011submodular} proved that no sub-exponential time algorithm can obtain a better than $0.478$-approximation for this problem even when the matroid is a partition matroid, and recently, Qi~\cite{qi2022maximizing} showed that the same inapproximability result applies even to the special case of a cardinality constraint. For the even more special case of unconstrained maximization of a (not necessarily monotone) submodular function, Feige et al.~\cite{feige11maximizing} showed that no sub-exponential time algorithm can obtain a better than $\nicefrac{1}{2}$-approximation, and Buchbinder et al.~\cite{buchbinder2015tight} gave a matching randomized approximation algorithm, which was later derandomized by Buchbinder \& Feldman~\cite{buchbinder2018deterministic}.

\paragraph{Paper Structure.} Section~\ref{sec:preliminaries} formally defines the problem we consider and the notation that we use, as well as introduces some known results that we employ. Section~\ref{sec:algorithm} analyzes Algorithm~\ref{alg:basic}, and then presents and analyzes the variant of this algorithm used to prove Theorem~\ref{thm:main}. The analysis in Section~\ref{sec:algorithm} employs Proposition~\ref{prop:fast}, whose proof appears in Section~\ref{sec:approximate-local}. The paper concludes with Appendix~\ref{app:submodular_linear_sums}, which shows how our technique can be used also for maximizing sums of a monotone submodular function and a linear function subject to a matroid constraint.
\section{Preliminaries} \label{sec:preliminaries}

In this section, we formally define the problem we consider in this paper. We also introduce the notation used in the paper and some previous results that we employ.

Let $f\colon 2^\cN \to \bR$ be a set function.
The function $f$ is \emph{non-negative} if $f(S) \geq 0$ for every set $S \subseteq \cN$, and it is \emph{monotone} if $f(S) \leq f(T)$ for every two sets $S \subseteq T \subseteq \cN$.
Recall that we use the shorthand $f(u \mid S) \triangleq f(S \cup \{u\}) - f(S)$ to denote the marginal contribution of element $u$ to the set $S$. Similarly, given two sets $S, T \subseteq \cN$, we use $f(T \mid S) \triangleq f(S \cup T) - f(S)$ to denote the marginal contribution of $T$ to $S$. A function $f$ is \emph{submodular} if $f(S) + f(T) \geq f(S \cap T) + f(S \cup T)$ for every two sets $S, T \subseteq \cN$. Following is another definition of submodularity, which is known to be equivalent to the above one.
\begin{definition}
A set function $f\colon 2^\cN \to \bR$ is submodular if $f(u \mid S) \geq f(u \mid T)$ for every two sets $S \subseteq T \subseteq \cN$ and element $u \in \cN \setminus T$.
\end{definition}

The next lemmas follow form the definition of submodularity.

\begin{lemma}\label{lem:submodularA}
Let $f\colon 2^\cN \to \bR$ be a submodular function, and let $S\subseteq T\subseteq \cN$. It holds that 
\[
    f(T)-f(S) \geq \sum_{u\in T\setminus S}f(u \mid T-u)
    \enspace.
\]
\end{lemma}

\begin{proof}
Denote the elements of $T\setminus S$ by $u_1, \ldots, u_k$. Then, by the submodularity of $f$,
\[f(T)-f(S) = \sum_{i=1}^{k}f(u_i \mid S \cup \{u_1, \ldots, u_{i-1}\}) \geq \sum_{u\in T\setminus S}f(u \mid T-u) \enspace. \qedhere\] 
\end{proof}

\begin{lemma}\label{lem:submodularB}
Let $f\colon 2^\cN \to \bR$ be a submodular function, and let $S, T\subseteq \cN$. It holds that 
\[
    \sum_{u\in T}f(u \mid S) \geq f(T\mid S)
    \enspace.
\]
\end{lemma}

\begin{proof}
Denote the elements of $T\setminus S$ by $u_1, \ldots, u_k$. Then, by the submodularity of $f$,
\[ \sum_{u\in T}f(u \mid S) = \sum_{u\in T\setminus S}\mspace{-9mu}f(u \mid S)\geq  \sum_{i=1}^{k}f(u_i \mid S \cup \{u_1, \ldots, u_{i-1}\}) = f(T\cup S)- f(S) = f(T\mid S)
\enspace.
\qedhere\]
\end{proof}

An independence system is a pair $(\cN, \cI)$, where $\cN$ is a ground set and $\cI \subseteq 2^\cN$ is a non-empty collection of subsets of $\cN$ that is down-closed in the sense that $T \in \cI$ implies that every set $S \subseteq T$ also belongs to $\cI$. Following the terminology used for linear spaces, it is customary to refer to the sets of $\cI$ as \emph{independent} sets. Similarly, independent sets that are inclusion-wise maximal (i.e., they are not subsets of other independent sets) are called \emph{bases}, and the size of the largest base is called the \emph{rank} of the independence system. Throughout the paper, we use $n$ to denote the size of the ground set $\cN$ and $r$ to denote the rank of the independence system.

A \emph{matroid} is an independence system that also obeys the following exchange axiom: if $S, T$ are two independent sets such that $|S| < |T|$, then there must exist an element $u \in T \setminus S$ such that $S \cup \{u\} \in \cI$. Matroids have been extensively studied as they capture many cases of interest, and at the same time, enjoy a rich theoretical structure (see, e.g.,~\cite{schrijver2003combinatorial}). However, we mention here only two results about them. The first result, which is an immediate corollary of the exchange axiom, is that all the bases of a matroid are of the same size (and thus, every independent set of this size is a base). The other result is given by the next lemma. In this lemma (and throughout the paper), given a set $S$ and element $u$, we use $S + u$ and $S - u $ as shorthands for $S \cup \{u\}$ and $S \setminus \{u\}$, respectively.
\begin{lemma}[Proved by~\cite{brualdi1969comments}, and can also be found as Corollary~39.12a in~\cite{schrijver2003combinatorial}] \label{le:perfect_matching_two_bases}
Let $A$ and $B$ be two bases of a matroid $\cM = (\cN, \cI)$. Then, there exists a bijection $h\colon A \setminus B \rightarrow B \setminus A$ such that for every $u \in A \setminus B$, $(B - h(u)) + u \in \cI$.
\end{lemma}
One can extend the domain of the function $h$ from the last lemma to the entire set $A$ by defining $h(u) = u$ for every $u \in A \cap B$. This yields the following corollary.
\begin{corollary} \label{cor:perfect_matching_two_bases}
Let $A$ and $B$ be two bases of a matroid $\cM = (\cN, \cI)$. Then, there exists a bijection $h\colon A \rightarrow B$ such that for every $u \in A$, $(B - h(u)) + u \in \cI$ and $h(u) = u$ for every $u \in A \cap B$.
\end{corollary}

In this paper, we study the following problem. Given a non-negative monotone submodular function $f\colon 2^\cN \to \nnR$ and a matroid $\cM = (\cN, \cI)$ over the same ground set, we would like to find an independent set of $\cM$ that maximizes $f$ among all independent sets of $\cM$. As is mentioned above, it is standard in the literature to assume access to the objective function $f$ and the matroid $\cM$ via two oracles. The first oracle is called \emph{value oracle}, and given a set $S \subseteq \cN$ returns $f(S)$. The other oracle is called \emph{independence oracle}, and given a set $S \subseteq \cN$ indicates whether $S \in \cI$.
When the function $f$ happens to be linear, a well-known greedy algorithm can be used to exactly solve the  problem we consider  using $O(n)$ value and independence oracle queries (see, e.g., Theorem~40.1 of~\cite{schrijver2003combinatorial}). Recall that our main result is a non-oblivious local search algorithm guaranteeing $(1 - 1/e - \eps)$-approximation for the general case of this problem. To determine whether this algorithm has successfully converged, we sometimes employ the algorithm given by the next lemma, which obtains a worse approximation ratio for the same problem.
\begin{lemma}[Lemma $3.2$ in~\cite{buchbinder2017comparing}, based on \cite{babanidiyuru2014fast}]\label{lem:fastapprox}
Given a non-negative monotone submodular function $f$ and a matroid $\cM = (\cN, \cI)$, there exists a deterministic $\nicefrac{1}{3}$-approximation algorithm for $\max\{f(S) \mid S\in \cI\}$ that has a query complexity of $O(n \log r)$.
\end{lemma}
\section{Algorithm} \label{sec:algorithm}

In this section, we present and analyze the algorithm used to prove our main result (Theorem~\ref{thm:main}). We do that in two steps. First, in Section~\ref{ssc:basic_algorithm}, we analyze Algorithm~\ref{alg:basic}, which is a simple idealized version of our final algorithm. Algorithm~\ref{alg:basic} does not run in polynomial time, but its analysis conveys the main ideas used in the analysis of our final algorithm. Then, in Section~\ref{ssc:implementation}, we show our final algorithm, which uses Proposition~\ref{prop:fast}, and has all the properties guaranteed by Theorem~\ref{thm:main}.

\subsection{Idealized Algorithm} \label{ssc:basic_algorithm}

In this section, we analyze the idealized version of our final algorithm that is given above as Algorithm~\ref{alg:basic}. In particular, we show that this idealized version guarantees $(1 - \nicefrac{1}{e} - \eps)$-approximation (but, unfortunately, its time complexity is not guaranteed to be polynomial). Recall that Algorithm~\ref{alg:basic} is parametrized by a positive integer value $\ell$ and additional positive values $\alpha_1, \ldots, \alpha_{\ell}$ to be determined later (we also define $\alpha_0 \triangleq 0$ for convenience). 
Algorithm~\ref{alg:basic} is a non-oblivious local search algorithm whose solution consists of $\ell$ disjoint sets $S_1, S_2, \ldots, S_{\ell}$ whose union is a base of $\cM$. The local search is guided by an auxiliary function 
\[g(S_1, \ldots, S_{\ell}) =  \sum_{J \subseteq [\ell]} \alpha_{|J|} \cdot f(S_J) \enspace,\]
where $S_{J} \triangleq \cupdot_{i\in J}S_j$ for any $J\subseteq [\ell]$. Intuitively, this objective function favors solutions that are non-fragile in the sense that removing a part of the solution has (on average) a relatively small effect on the solution's value.


The output set of Algorithm~\ref{alg:basic} is the set $S_{[\ell]}$, which is clearly a base of the matroid $\cM$. The rest of this section is devoted to analyzing the approximation guarantee of Algorithm~\ref{alg:basic}. We begin with the following simple observation, which follows from the fact that Algorithm~\ref{alg:basic} terminates with a local maximum with respect to $g$. In this observation, and all the other claims in this section, the sets mentioned represent their final values, i.e., their values when Algorithm~\ref{alg:basic} terminates.
\begin{observation} \label{obs:local_minimum}
Let $j,k\in [\ell]$. For every $u\in S_k$, it holds that
\begin{equation}
    \sum_{J\subseteq [\ell], k \in J, j \not \in J} \mspace{-27mu} \alpha_{|J|} \cdot f(u \mid S_J-u) \geq \sum_{J\subseteq [\ell], k \not \in J, j \in J} \mspace{-27mu} \alpha_{|J|} \cdot f(u \mid S_J) \enspace. \label{ineq112}
\end{equation}
Furthermore, for any $v \in \cN \setminus S_{[\ell]}$ such that $S_{[\ell]}-u+v\in \cI$, it also holds that
 \begin{equation}
\sum_{J \subseteq [\ell], k \in J} \mspace{-9mu} \alpha_{|J|} \cdot f(u \mid S_J-u) \geq \sum_{J \subseteq [\ell], j \in J} \mspace{-9mu} \alpha_{|J|} \cdot f(v \mid S_J) \enspace. \label{ineq111}
\end{equation}   
\end{observation}
\begin{proof}
Inequality~\eqref{ineq112} is an immediate consequence of the fact that, when  Algorithm~\ref{alg:basic} terminates, Step~1 cannot increase $g$ by moving $u$ from $S_k$ to $S_j$. Such a change removes $u$ from all $S_J$ such that $k\in J$ and $j\notin J$ and adds $u$ to all $S_J$ such that $k\notin J$ and $j\in J$. The LHS measures the decrease in $g$ due to the first part, while the RHS measures the increase of $g$ due to the second part. As this move does not increase $g$, the RHS is at most the LHS. 

Similarly, Inequality~\eqref{ineq111} follows since Step~2 of Algorithm~\ref{alg:basic} cannot increase $g$ by removing $u$ from $S_k$ and adding $v\in \cN \setminus S_{[\ell]}$ to $S_j$. Removing $u$ from $S_k$ decreases $g$ by the LHS of Inequality~\eqref{ineq111}. Then, adding $v$ to $S_j$ increases $g$ by 
\begin{align*}
   \sum_{J \subseteq [\ell], j \in J, k\in J}  \alpha_{|J|} \cdot f(v \mid S_J-u) + \sum_{J \subseteq [\ell], j \in J, k\notin J}  \alpha_{|J|} \cdot f(v \mid S_J) \geq \sum_{J \subseteq [\ell], j \in J}  \alpha_{|J|} \cdot f(v \mid S_J) 
\end{align*}
where the last inequality holds by the submodularity of $f$. As this valid option of Step~2 cannot increase $g$, the RHS of Inequality~\eqref{ineq111} is at most its LHS.
%
\end{proof}


The next lemma uses the previous observation to get an inequality relating the value of a linear combination of the values of the sets $S_J$ to the value of the optimal solution. Note that since $f$ is monotone, there must be an optimal solution for our problem that is a base of $\cM$. Let us denote such an optimal solution by $OPT$. 

\begin{lemma}\label{lem:main}
Let $\alpha_{\ell + 1} \triangleq 0$. Then,
\[\sum_{J \subseteq [\ell]} 
\Big[\alpha_{|J|} \cdot |J|\cdot \Big(1+\frac{1}{\ell}\Big) -\alpha_{|J|+1} \cdot (\ell-|J|)\Big] \cdot f(S_J) \geq \sum_{i=1}^{\ell}\binom{\ell-1 }{i-1} \cdot \alpha_{i} \cdot f(OPT)\enspace.\]
\end{lemma}

\begin{proof}
Since $OPT$ and $S_{[\ell]}$ are both bases of $\cM$, Corollary~\ref{cor:perfect_matching_two_bases} guarantees that there exists a bijective function $h\colon S_{[\ell]} \to OPT$ such that $S_{[\ell]} - u + h(u) \in \cI$ for every $u \in S_{[\ell]}$ and $h(u) = u$ for every $u \in S_{[\ell]} \cap OPT$.
We claim that the following inequality holds for every element $u \in S_k$ and integer $j\in[\ell]$. 
\begin{equation} \label{ineq111-like}
\sum_{J \subseteq [\ell], k \in J} \mspace{-9mu} \alpha_{|J|} \cdot f(u \mid S_J-u) \geq \sum_{J \subseteq [\ell], j \in J} \mspace{-9mu} \alpha_{|J|} \cdot f(h(u) \mid S_J)
\enspace.
\end{equation}
If $u \not \in OPT$, then the last inequality is an immediate corollary of Inequality~\eqref{ineq111} since $h(u) \not \in S_{[\ell]}$. Thus, we only need to consider the case of $u \in OPT$. To handle this case, we notice that the monotonicity of $f$ implies that
\[
 \sum_{J\subseteq [\ell], k \in J, j \in J} \mspace{-27mu} \alpha_{|J|} \cdot f(u \mid S_J-u) \geq 0 = \sum_{J\subseteq [\ell], k \in J, j \in J} \mspace{-27mu} \alpha_{|J|} \cdot f(u \mid S_J)
 \enspace,
\]
and adding this inequality to Inequality~\eqref{ineq112} yields
\[
\sum_{J \subseteq [\ell], k \in J} \mspace{-9mu} \alpha_{|J|} \cdot f(u \mid S_J-u) \geq \sum_{J \subseteq [\ell], j \in J} \mspace{-9mu} \alpha_{|J|} \cdot f(u \mid S_J)
\enspace,
\]
which implies Inequality~\eqref{ineq111-like} since $h(u) = u$ when $u \in S_k \cap OPT \subseteq S_{[\ell]} \cap OPT$.

Averaging Inequality~\eqref{ineq111-like} over all $j\in [\ell]$, we get
\[
\sum_{J \subseteq [\ell], k \in J} \mspace{-9mu} \alpha_{|J|} \cdot f(u \mid S_J-u) \geq  \sum_{J \subseteq [\ell]} \frac{|J|}{\ell} \cdot \alpha_{|J|} \cdot f(h(u) \mid S_J) \enspace,
\label{ineq-m}
\]
and adding this inequality up over all $u \in S_{[\ell]}$ implies

\begin{align*}
\sum_{J \subseteq [\ell]} 
\left[\alpha_{|J|} \cdot |J| -\alpha_{|J|+1} \cdot (\ell-|J|)\right] \cdot f(S_J) &= 
\sum_{J \subseteq [\ell]}\sum_{ k \in J} 
\alpha_{|J|} \cdot [f(S_J) -f(S_{J\setminus \{k\}})]\\
& \geq \sum_{J \subseteq [\ell] } \sum_{ k \in J}   \alpha_{|J|} \cdot \Big(\sum_{u\in S_k} f(u \mid S_J-u)\Big) \\
& = \sum_{k=1}^{\ell}\sum_{u\in S_k}\sum_{J \subseteq [\ell], k \in J} \mspace{-9mu}  \alpha_{|J|} \cdot f(u \mid S_J-u) \\
& \geq \sum_{k=1}^{\ell}\sum_{u\in S_k}\sum_{J \subseteq [\ell]} \frac{|J|}{\ell} \cdot \alpha_{|J|} \cdot f(h(u) \mid S_J) \\
& =  \sum_{J \subseteq [\ell]} \frac{|J|}{\ell} \cdot \alpha_{|J|} \cdot \Big( \sum_{u\in S_{[\ell]}} f(h(u) \mid S_J)\Big) \\
& \geq \sum_{J \subseteq [\ell]}\frac{|J|}{\ell} \cdot \alpha_{|J|} \cdot f(OPT \mid S_J) \\ 
& \geq \sum_{J \subseteq [\ell]}\frac{|J|}{\ell} \cdot \alpha_{|J|} \cdot \left[f(OPT) -f(S_J)\right] 
\enspace,
\end{align*}
where the first inequality holds by Lemma \ref{lem:submodularA}, the 
penultimate inequality holds by Lemma \ref{lem:submodularB}, and the last inequality follows from $f$'s monotonicity. The lemma now follows by rearranging the last inequality since
\[\sum_{J \subseteq [\ell]}\frac{|J|}{\ell} \cdot \alpha_{|J|} \cdot f(OPT) = \sum_{i=1}^{\ell}\binom{\ell }{i} \cdot \frac{i}{\ell} \cdot \alpha_{i} \cdot f(OPT) = \sum_{i=1}^{\ell}\binom{\ell-1 }{i-1} \cdot \alpha_{i} \cdot f(OPT)
\enspace.
\qedhere
\]

\end{proof}

To get an approximation guarantee for Algorithm~\ref{alg:basic} from the last lemma, we need to make the coefficient of $f(S_J)$ zero for every $J \not \in \{\varnothing, [\ell]\}$. Setting $\alpha_1=1$, which is simply a scaling, gives a recurrence relation for choosing appropriate values for the parameters $\alpha_1, \alpha_2, \dotsc, \alpha_\ell$ leading to the following proposition.   
\propApproximationBasic*
\begin{proof}
The second part of the proposition follows from the first part by the non-negativity of $f$ and the observation that, for $\ell \geq 2$, we have
\[
    (1 + 1/\ell)^{-\ell}
    \leq
	\frac{1}{e(1 - 1/\ell)} \leq \frac{1}{e}(1+\frac{2}{\ell}) = \frac{1}{e} + O(\eps)
    \enspace.
\]
Thus, we concentrate on proving the first part of the theorem.

The values we have chosen for $\alpha_i$ imply that, for every $i \in [\ell - 1]$,
\begin{align*}
    \alpha_{i+1} & = \frac{(1 + 1/\ell)^{i}}{\binom{\ell - 1}{i}}  = (1+\frac{1}{\ell})\cdot \frac{(1 + 1/\ell)^{i - 1}}{\binom{\ell - 1}{i - 1} \cdot \frac{\ell-i}{i}}  = (1+\frac{1}{\ell}) \cdot \frac{i}{\ell-i} \cdot \alpha_i \enspace.
\end{align*}
Thus, for every $i \in [\ell - 1]$, we have the equality $\alpha_i \cdot i (1 + 1/\ell) - \alpha_{i + 1}(\ell - i) = 0$, and plugging this equality into Lemma~\ref{lem:main} yields
\begin{align*}
    \alpha_\ell \cdot (\ell + 1) \cdot f(S_{[\ell]}) - \alpha_1 \cdot \ell \cdot f(\varnothing)
    ={} &
\sum_{J \subseteq [\ell]} \Big[\alpha_{J} \cdot |J| \cdot \Big(1+\frac{1}{\ell}\Big) -\alpha_{|J|+1} \cdot (\ell - |J|)\Big] \cdot f(S_{[\ell]}) \\
    \geq{} &
    \sum_{i=1}^{\ell} (1 + 1/\ell)^{i - 1} \cdot f(OPT)
    =
    \ell \cdot [(1 + 1/\ell)^\ell - 1] \cdot f(OPT)
    \enspace.
\end{align*}
The proposition now follows by plugging the values $\alpha_\ell = (1 + 1/\ell)^{\ell - 1}$ and $\alpha_1 = 1$ into the last inequality, and rearranging.
%
%
 \end{proof}
 
\subsection{Final Algorithm}\label{ssc:implementation}

In this section, we describe the final version of our algorithm, which is similar to the idealized version (Algorithm~\ref{alg:basic}) studied in Section~\ref{ssc:basic_algorithm}, but uses Proposition~\ref{prop:fast} to get the query complexities stated in Theorem~\ref{thm:main}.

Algorithm~\ref{alg:basic} maintains $\ell$ disjoint sets $S_1, S_2, \dotsc, S_\ell$, and accordingly, the auxiliary function $g$ used to guide it accepts these $\ell$ sets as parameters. To use Proposition~\ref{prop:fast}, we need to replace this auxiliary function with a function accepting only a single parameter. The first step towards this goal is defining a new ground set $\cN' \triangleq \cN \times [\ell]$. Intuitively, every pair $(u, i)$ in this ground set represents the assignment of element $u$ to $S_i$. Given this intuition, we get a new auxiliary function $g' \colon 2^{\cN'} \to \nnR$ as follows. For every set $S \subseteq \cN'$ and for every $J\subseteq [\ell]$, let $\pi_J(S) \triangleq \{u \in \cN \mid \exists_{i \in J}\; (u, i) \in S\}$ be the set of elements that intuitively should be in the set $\cup_{i \in J} S_i$ according to $S$. Then, for every $S \subseteq \cN'$,
\[
	g'(S)
	\triangleq
	g\left(\pi_1(S), \ldots, \pi_{\ell}(S)\right)
	=
	\sum_{J\subseteq [\ell]}\alpha_{|J|} \cdot f(\pi_J(S))
	\enspace.
\]
We also need to lift the matroid $\cM$ to a new matroid over the ground set $\cN'$. Since the sets $S_1, S_2, \dotsc, S_\ell$ should be kept disjoint, for every $u \in \cN$, the elements of $\{(u, i) \mid i \in [\ell]\}$ should be parallel elements in the new matroid.\footnote{Elements $u$ and $v$ of a matroid are parallel if (i) they cannot appear together in any independent set of the matroid, and (ii) for every independent set $S$ that contains $u$ (respectively, $v$), the set $S - u + v$ (respectively, $S - v + u$) is also independent.} Thus, it is natural to define the new matroid as $\cM' = (\cN', \cI')$, where $\cI' \triangleq \{S \subseteq \cN' \mid \forall_{u \in \cN}\; |S \cap (\{u\} \times [\ell])| \leq 1 \text{ and } \pi_{[\ell]}(S) \in \cI\}$ (we prove below that $\cM'$ is indeed a matroid).

Notice that a value oracle query to $g'$ can be implemented using $2^\ell$ value oracle queries to $f$, and an independence oracle query to $\cM'$ can be implemented using a single independence oracle query to $\cM$. The following two lemmata prove additional properties of $g'$ and $\cM'$, respectively.
\begin{lemma}
The function $g'$ is non-negative, monotone and submodular.
\end{lemma}
\begin{proof}
%
The non-negativity and monotonicity of $g'$ follow immediately from its definition and the fact that $f$ has these properties. Thus, we concentrate on proving that $g$ is submodular.
For every two sets $S \subseteq T \subseteq \cN'$ and element $v=(u,i)\not\in T$,
\begin{align*}
	g'(v \mid S) 
	= 
	\sum_{J\subseteq [\ell], i \in J} \mspace{-9mu} \alpha_{|J|} \cdot f(u \mid \pi_J(S)) 
	\geq
	\sum_{J\subseteq [\ell], i \in J} \mspace{-9mu} \alpha_{|J|} \cdot f(u \mid \pi_J(T))
	= g'(v \mid T)
	\enspace,
\end{align*}
where the inequality holds separately for each set $J$. If $u\not \in \pi_J(T)$, then the inequality holds for this $J$ by the submodularity of $f$, and otherwise, it holds by the monotonicity of $f$.
\end{proof}

\begin{lemma}
The above defined $\cM'$ is a matroid, and its has the same rank as $\cM$.
\end{lemma}
\begin{proof}
The fact that $\cM'$ is an independence system follows immediately from its definition and the fact that $\cM$ is an independence system. Thus, to prove that $\cM'$ is a matroid, it suffices to show that it obeys the exchange axiom of matroids. Consider two sets $S, T \in \cI'$ such that $|S| < |T|$. The membership of $S$ and $T$ in $\cI'$ guarantees that $\pi_{[\ell]}(S), \pi_{[\ell]}(T) \in \cI$, $|\pi_{[\ell]}(S)| = |S|$ and $|\pi_{[\ell]}(T)| = |T|$. Thus, since $\cM$ is a matroid, there must exist an element $u \in \pi_{[\ell]}(T) \setminus \pi_{[\ell]}(S)$ such that $\pi_{[\ell]}(S) + u \in \cI$. By the definition of $\pi_{[\ell]}$, the membership of $u$ in $\pi_{[\ell]}(T)$ implies that there is an index $i \in [\ell]$ such that $(u, i) \in T$. Additionally, $(u, i) \not \in S$ since $u \not \in \pi_{[\ell]}(S)$, and one can verify that $S + (u, i) \in \cI'$. This completes the proof that $\cM'$ obeys the exchange axiom of matroids.

It remains to prove that $\cM'$ has the same rank as $\cM$. On the one hand, we notice that for every set $S \in \cI$, we have that $S \times \{1\} \in \cI'$. Thus, the rank of $\cM'$ is not smaller than the rank of $\cM$. On the other hand, consider a set $S \in \cI'$. The fact that $S \in \cI'$ implies that $|\pi_{[\ell]}(S)| = |S|$ and $\pi_{[\ell]}(S) \in \cI$, and thus, the rank of $\cM$ is also not smaller than the rank of $\cM'$.
\end{proof}

We are now ready to present our final algorithm, which is given as Algorithm~\ref{alg:final}. This algorithm has a single positive integer parameter $\ell$. As mentioned in Section~\ref{ssc:contribution}, the set $S$ obtained by Algorithm~\ref{alg:final} is not a local maximum of $g'$ in the same sense that the output set of Algorithm~\ref{alg:basic} is a local maximum of $g$ (see Remark \ref{rem:localOPT}). Nevertheless, Lemma~\ref{lemma:approximation_final} shows that the set $S$ obeys basically the same approximation guarantee (recall that $OPT$ is a feasible solution maximizing $f$).

\begin{algorithm}[ht]
\caption{\textsc{Fast Non-Oblivious Local Search}} \label{alg:final}
Define the ground set $\cN'$, the function $g'$ and the matroid $\cM'$ as described in Section~\ref{ssc:implementation}. In the function $g'$, set the parameters $\alpha_1, \alpha_2, \dotsc, \alpha_\ell$ as in Proposition~\ref{prop:approximation_basic}.\\
Let $\eps' \gets \eps / (e(1 + \ln \ell))$.\\
Use either the deterministic or the randomized algorithm from Proposition~\ref{prop:fast} (with $S_0$ set to $\varnothing$) to find a set $S \in \cI'$ that with probability at least $1 - \eps'$ obeys
\begin{equation} \label{eq:final_alg_guarantee}
	\sum_{v\in T}g'(v \mid S - v)- \sum_{u\in S}g'(u \mid S-u) \leq \eps' \cdot g'(OPT') \qquad \forall T\in \cI'
	\enspace,
\end{equation}
where $OPT'$ is a set of $\cI'$ maximizing $g'$ (the term $g'(S_0)$ from the guarantee of Proposition~\ref{prop:fast} was dropped from the last inequality since it is non-negative). 
\label{line:maximizing_g_prime}

\Return $\pi_{[\ell]}(S)$.
\end{algorithm}

\begin{lemma} \label{lemma:approximation_final}
If a set $S \subseteq \cN'$ obeys Inequality~\eqref{eq:final_alg_guarantee}, then
\[
	f(\pi_{[\ell]}(S))
	\geq
	\big(1 - (1 + \nicefrac{1}{\ell})^{-\ell}\big)\cdot f(OPT) + (1 + \nicefrac{1}{\ell})^{-\ell} \cdot f(\varnothing) - \eps \cdot f(OPT)
  \enspace.
\]
\end{lemma}
\begin{proof}
For every $j \in [\ell]$, we have $OPT \times \{j\} \in \cI'$. Thus, Inequality~\eqref{eq:final_alg_guarantee} implies that
\begin{align*}
	\sum_{u \in OPT} \sum_{J \subseteq [\ell], j \in J} \alpha_{|J|} \cdot f(&u \mid \pi_J(S))
	=
	\sum_{u \in OPT} \sum_{J\subseteq [\ell]}\alpha_{|J|} \cdot [f(\pi_J(S + (u, j))) - f(\pi_J(S))]\\
	={} &
	\sum_{u \in OPT} g'((u, j) \mid S)
	\leq
	\sum_{u \in OPT} g'((u, j) \mid S - (u, j))\\
	\leq{} &
	\sum_{(u, k) \in S} g'((u, k) \mid S- (u, k)) + \eps' \cdot g'(OPT')\\
	={} &
	\sum_{(u, k) \in S} \sum_{J \subseteq [\ell], k \in J} \mspace{-9mu}\alpha_{|J|} \cdot f(u \mid \pi_J(S) - u) + \eps' \cdot g'(OPT')
	\enspace,
\end{align*}
where the first inequality holds by the monotonicity of $g'$. Averaging this inequality over all $j \in [\ell]$ now yields
\[
	\sum_{(u, k) \in S} \sum_{J \subseteq [\ell], k \in J} \mspace{-9mu} \alpha_{|J|} \cdot f(u \mid \pi_J(S)-u)
	\geq
	\sum_{u \in OPT} \sum_{J \subseteq [\ell]} \frac{|J|}{\ell} \cdot \alpha_{|J|} \cdot f(u \mid \pi_J(S)) -\eps' \cdot g'(OPT')
	\enspace,
\]
which implies 
{\allowdisplaybreaks \begin{align*}
\sum_{J \subseteq [\ell]} 
\left[\alpha_{|J|} \cdot |J| -\alpha_{|J|+1} \cdot (\ell-|J|)\right] \cdot f&(\pi_J(S)) = 
\sum_{J \subseteq [\ell]}\sum_{k\in J}
\alpha_{|J|} \cdot [f(\pi_J(S)) -f(\pi_{J\setminus \{k\}}(S))]\\
& \geq \sum_{J \subseteq [\ell]}\sum_{k\in J} \alpha_{|J|} \cdot \Big(\sum_{(u, k)\in S} f(u \mid \pi_J(S)-u)\Big) \\
& = \sum_{(u, k)\in S} \sum_{J \subseteq [\ell], k \in J} \mspace{-9mu}  \alpha_{|J|} \cdot f(u \mid \pi_J(S)-u)\\
& \geq \sum_{u \in OPT} \sum_{J \subseteq [\ell]} \frac{|J|}{\ell} \cdot \alpha_{|J|} \cdot f(u \mid \pi_J(S)) -\eps' \cdot g'(OPT') \\
& \geq \sum_{J \subseteq [\ell]}\frac{|J|}{\ell} \cdot \alpha_{|J|} \cdot f(OPT \mid \pi_J(S)) -\eps' \cdot g'(OPT') \\ 
& \geq \sum_{J \subseteq [\ell]}\frac{|J|}{\ell} \cdot \alpha_{|J|} \cdot \left[f(OPT) -f(\pi_J(S))\right] -\eps' \cdot g'(OPT') 
\enspace.
\end{align*}}%
where the first inequality holds by Lemma \ref{lem:submodularA}, the 
penultimate inequality holds by Lemma \ref{lem:submodularB}, and the last inequality follows from $f$'s monotonicity.

The final inequality we get is the same as in the proof of Lemma~\ref{lem:main}, except for the error term $-\eps' \cdot g'(OPT')$ and a replacement of $S_J$ with $\pi_J(S)$. Continuing with the proof of Lemma~\ref{lem:main} from this point, we get that this lemma still applies up to the same modifications (when the condition of the current lemma holds), and plugging this observation into the proof of Proposition~\ref{prop:approximation_basic} then yields
\begin{equation} \label{eq:raw_final_approximation}
	\alpha_\ell \cdot (\ell + 1) \cdot f(\pi_{[\ell]}(S)) - \alpha_1 \cdot \ell \cdot f(\varnothing)
  \geq
  \ell \cdot [(1 + 1/\ell)^\ell - 1] \cdot f(OPT) -\eps' \cdot g'(OPT')
  \enspace.
\end{equation}
Recall that $\alpha_\ell = (1 + 1/\ell)^{\ell - 1}$ and $\alpha_1 = 1$. Additionally, since $\pi_J(OPT') \subseteq \pi_{[\ell]}(OPT') \in \cI$ for every set $J \subseteq [\ell]$, the definition of $OPT$ implies that
\begin{align*}
	g'(OPT')
	={} &
	\sum_{J\subseteq [\ell]}\alpha_{|J|} \cdot f(\pi_J(OPT'))
	\leq
	\sum_{J\subseteq [\ell]}\alpha_{|J|} \cdot f(OPT)\\
	={} &
	\sum_{i = 1}^\ell \frac{\ell \cdot (1 + 1/\ell)^{i - 1}}{i} \cdot f(OPT)
	\leq
	e\ell \cdot \sum_{i = 1}^\ell \frac{1}{i} \cdot f(OPT)
	\leq
	e\ell (1 + \ln \ell) \cdot f(OPT)
	\enspace.
\end{align*}
The lemma now follows by plugging all the above observations into Inequality~\eqref{eq:raw_final_approximation}, and rearranging, since
\[
	\frac{\eps' \cdot g'(OPT')}{\alpha_\ell \cdot (\ell + 1)}
	=
	\frac{\eps \cdot g'(OPT') / (e(1 + \ln \ell))}{(1 + 1/\ell)^{\ell - 1} \cdot (\ell + 1)}
	\leq
	\frac{\eps\ell \cdot f(OPT)}{\ell + 1}
	\leq
	\eps \cdot f(OPT)
	\enspace.
	\qedhere
\]
\end{proof}
\begin{corollary} \label{cor:approximation_final}
It always holds that $\bE[f(\pi_{[\ell]}(S))] \geq (1 - (1 + \nicefrac{1}{\ell})^{-\ell})\cdot f(OPT) - O(\eps) \cdot f(OPT)$. In particular, 
for every $\eps > 0$, if $\ell$ is set to $1 + \lceil 1/\eps \rceil$ in Algorithm~\ref{alg:final}, then $\bE[f(\pi_{[\ell]}(S))] \geq\left(1 - \nicefrac{1}{e} - O(\eps)\right) \cdot f(OPT)$.
\end{corollary}
\begin{proof}
Lemma~\ref{lemma:approximation_final} implies that the output set $\pi_{[\ell]}(S)$ of Algorithm~\ref{alg:final} obeys, with probability at least $1 - \eps' \geq 1 - \eps$, the inequality
\[
	f(\pi_{[\ell]}(S))
	\geq
	\big(1 - (1 + \nicefrac{1}{\ell})^{-\ell}\big)\cdot f(OPT) + (1 + \nicefrac{1}{\ell})^{-\ell} \cdot f(\varnothing) - \eps \cdot f(OPT)
  \enspace.
\]
Since the non-negativity of $f$ guarantees that $f(\pi_{[\ell]}(S)) \geq 0$ and $f(\varnothing) \geq 0$, the above inequality implies
\begin{align*}
	\bE[f(\pi_{[\ell]}(S))]
	\geq{} &
	(1 - \eps) \cdot \big[\big(1 - (1 + \nicefrac{1}{\ell})^{-\ell}\big)\cdot f(OPT) + (1 + \nicefrac{1}{\ell})^{-\ell} \cdot f(\varnothing) - \eps \cdot f(OPT)\big]\\
	\geq{} &
	(1 - \eps) \cdot \big(1 - (1 + \nicefrac{1}{\ell})^{-\ell}\big)\cdot f(OPT) - \eps \cdot f(OPT)\\
	\geq{} &
	\big(1 - (1 + \nicefrac{1}{\ell})^{-\ell}\big)\cdot f(OPT) - 2\eps \cdot f(OPT)
	\enspace.
\end{align*}

The corollary now follows since the calculations done in the first part of the proof of Proposition~\ref{prop:approximation_basic} guarantee that, for $\ell = 1 + \lceil 1/\eps \rceil$, $(1 + 1/\ell)^{-\ell} = \nicefrac{1}{e} + O(\eps)$.
\end{proof}

We next complete the proof of Theorem~\ref{thm:main}.

\begin{proof}[Proof of Theorem ~\ref{thm:main}]

Assume that we choose for $\ell$ the value $1 + \lceil 1/ \eps \rceil$ as in Corollary~\ref{cor:approximation_final}. To prove the theorem, we only need to show that for this choice of value for $\ell$, Algorithm~\ref{alg:final} obeys the query complexities stated in the theorem. Below, we show that the query complexity of the deterministic version of Algorithm~\ref{alg:final} is $\tilde{O}_\eps(nr)$, and the query complexity of the randomized version is $\tilde{O}_\eps(n + r\sqrt{n})$.
    
Note that the above choice of value for $\ell$ implies that both $\ell$ and $1/\eps'$ can be upper bounded by functions of $\eps$. Thus, a value oracle query to $g'$ can be implemented using $2^\ell= O_{\eps}(1)$ value oracle queries to $f$, and the size of the ground set of $\cM'$ is $n\cdot \ell= O_{\eps}(n)$. Recall also that an independence oracle query to $\cM'$ can be implemented using a single independence oracle query to $\cM$, and the rank of $\cM'$ is $r$. These observations imply together that the query complexities of the deterministic and randomized versions of Algorithm~\ref{alg:final} are identical to the query complexities stated in Proposition~\ref{prop:fast} when one ignores terms that depend only on $\eps$. The lemma now follows since the deterministic algorithm of Proposition~\ref{prop:fast} requires $\tilde{O}(\eps^{-1}nr) = \tilde{O}_\eps(nr)$ queries, and the randomized algorithm of Proposition~\ref{prop:fast} requires $\tilde{O}(\eps^{-1}(n+r\sqrt{n})) = \tilde{O}_\eps(n + r\sqrt{n})$ queries.
\end{proof}



\section{Fast Algorithms for Finding a Relaxed Local Optimum}\label{sec:approximate-local}


In this section, we prove Proposition~\ref{prop:fast}, which we repeat here for convenience.
{\let\forcelabel\label \let\label\ignore \renewcommand{\theproposition}{\ref*{prop:fast}} \propFast \addtocounter{theorem}{-1}}

\begin{remark}\label{rem:localOPT}
We note that a set $S$ obeying the guarantee \eqref{ineq-localOPT} in Proposition~\ref{prop:fast} is not a local maximum of $f$ in the usual sense. A true (approximate) local optimum would be an independent set $S \in \cI$ that obeys: (i) $f(S) \geq f(S - u) - (\eps/r) \cdot [f(OPT) - f(S_0)]$ for every $u \in S$, (ii) $f(S ) \geq f(S + v) - (\eps/r) \cdot [f(OPT) - f(S_0)]$ for every $v \not \in S$ for which $S + v$ is independent in the matroid, and (iii) $f(S) \geq f(S - u + v) - (\eps / r) \cdot [f(OPT) - f(S_0)]$ for every $u \in S$ and $v \in \cN \setminus S$ for which $S - u + v$ is independent.\footnote{Notice that when $f$ is monotone, condition (i) is redundant, and condition (ii) is weaker than requiring $S$ to be a base of the matroid. Thus, the output of our idealized Algorithm~\ref{alg:basic} is indeed an approximate local optimum according to this definition.} We claim that such a local optimum always obeys \eqref{ineq-localOPT}. To see that, fix a set $T \in \cI$, and let $S'$ and $T'$ be arbitrary bases of $\cM$ that include $S$ and $T$ as subsets, respectively. Let $h\colon T' \to S'$ be the bijective function whose existence is guaranteed by Corollary~\ref{cor:perfect_matching_two_bases} for the bases $S'$ and $T'$. Then, by the submodularity of $f$,
\begin{align*}
	\sum_{v\in T}f(v \mid S - v)-{}& \sum_{u\in S}f(u \mid S - u)\\
    ={} &
    \sum_{\substack{v \in T \\ h(v) \not \in S}} \mspace{-9mu} f(v \mid S) - \sum_{\substack{u \in S \\ h^{-1}(u) \not \in T}} \mspace{-18mu} f(u \mid S - u) + \sum_{\substack{v \in T \\ h(v) \in S}} \mspace{-9mu} [f(v \mid S - v) - f(h(v) \mid S - h(v))]\\
    \leq{} &
    \sum_{\substack{v \in T \\ h(v) \not \in S}} \mspace{-9mu} f(v \mid S) - \sum_{\substack{u \in S \\ h^{-1}(u) \not \in T}} \mspace{-18mu} f(u \mid S - u) + \sum_{\substack{v \in T \\ h(v) \in S}} \mspace{-9mu} [f(S - h(v) + v) - f(S)]\\
    \leq{} &
    |T'| \cdot (\eps / r) \cdot [f(OPT) - f(S_0)]
    \leq
    \eps \cdot [f(OPT) - f(S_0)]
    \enspace,
\end{align*}
where the penultimate inequality holds since the number of terms in the second sum is at most $|T' \setminus T|$.

This shows that the guarantee in Proposition~\ref{prop:fast} is indeed weaker than a true (approximate) local optimum. However, as shown in Section~\ref{ssc:implementation}, this weaker guarantee suffices for our purposes, and in this section we describe faster ways to obtain it.
\end{remark}

The deterministic and randomized algorithms mentioned by Proposition~\ref{prop:fast} are presented in Sections~\ref{ssc:deterministic} and~\ref{ssc:randomized}, respectively. Both algorithms require the following procedure, which uses the binary search technique suggested by~\cite{chakrabarty2019faster,nguyen2019note}.

\begin{lemma}\label{lem:independence}
Let $S$ be an independent set of the matroid $\cM$, $S'$ be a subset of $S$, and $v$ be an element of the ground set $\cN$ such that $S + v \not \in \cI$, but $S \setminus S' + v \in \cI$. Then, there exists an element $u\in S'$ such that $S+v-u\in\cI$, and furthermore, there is an algorithm that gets as input a weight $w_u$ for every element $u\in S'$, and finds an element in $\arg \min_{u\in S', S+v-u\in\cI}\{w_u\}$ using $O(\log |S'|)$ independence oracle queries. 
\end{lemma}
\begin{proof}
    Let us denote, for simplicity, the elements of $S'$ by the numbers $1$ up to $|S'|$ in a non-decreasing weight order (i.e., $w_1 \leq w_2 \leq \ldots \leq w_{|S'|}$). Since $S + v \not \in \cI$, but $S \setminus S' + v \in \cI$, the down-closedness of $\cI$ implies that there exist an element $j \in [|S'|]$ such that $(S \setminus S') \cup \{v, i, i+1 \ldots, |S'|\}\in \cI$ for all $i>j$, and $(S \setminus S') \cup \{v, i, i+1, \ldots, |S'|\}\not\in \cI$ for $i\leq j$. The element $j$ can be found via binary search using $O(\log |S'|)$ independence oracle queries.

		We claim that the above element $j$ is a valid output for the algorithm promised by the lemma. To see this, note that since $(S \setminus S') \cup \{v, j+1, j+2 \ldots, |S'|\}\in \cI$, the matroid exchange axiom guarantees that it is possible to add elements from $S$ to the last set to get a subset of $S + v$ of size $|S|$ that is independent. Since $(S \setminus S') \cup \{v, j, j+1 \ldots, |S'|\}\not \in \cI$, this independent subset cannot include $j$, and thus, must be $S - j + v$. In contrast, for every $i < j$, $S - i + v \supseteq (S \setminus S') \cup \{v, j, j+1, \ldots, |S'|\}\not\in \cI$, and thus, $S - i + v \not \in \cI$.
	\end{proof}

To simplify the exposition of our algorithms, we implicitly assume in the rest of this section that the ground set $\cN$ includes a set $D$ of $r$ dummy elements. These dummy elements do not affect the objective function, and adding them to an independent set of the matroid keeps the set independent, unless its size exceeds $r$. More formally, for every set $S \subseteq \cN$,
\begin{itemize}
    \item $f(S) = f(S \setminus D)$, and
    \item $S \in \cI$ if and only if $S \setminus D \in \cI$ and $|S| \leq r$.
\end{itemize}
If the ground set $\cN$ does not naturally contain such a set $D$ of dummy elements, then $r$ dummy elements can be artificially added to the ground set before the execution of our algorithms. In this case, the objective function and the matroid should be extended to sets that include the new dummy elements using the above formal properties of the dummy elements. One can verify that these extensions keep the matroid legal, and also preserve the submodularity and non-negativity of the objective function. Furthermore, each value or independence oracle query to the extended objective function or matroid can be evaluated using a single value or independence oracle query to the original objective function or matroid, and therefore, the introduction of the dummy elements does not affect the query complexity of our algorithms. Finally, we note that if dummy elements are artificially added to the ground set, then any such elements that end up in the output set of our algorithms have to be removed from these sets after the algorithms terminate. Fortuanately, however, this removal does not affect Inequality~\eqref{ineq-localOPT}.

The idea of introducing dummy elements to simplify the exposition of algorithms for submodular maximization can be traced back to~\cite{buchbinder2014submodular}. In our algorithms, these elements are useful since they allow us to restrict our attention to bases of the matroid and to swap operations that add one element and remove another. When the function $f$ is monotone, one can make these assumptions even in the absence of dummy elements, and thus, the algorithms we present in this section work for such functions even if dummy elements are not introduced.

\subsection{A Deterministic Algorithm} \label{ssc:deterministic}

In this section, we describe and analyze the deterministic algorithm promised by Proposition~\ref{prop:fast}, which appears as Algorithm~\ref{alg:fastlocal}. Lemma~\ref{lem:deterministic} proves the correctness of this algorithm and the bounds stated in Proposition~\ref{prop:fast} on the the number of queries it makes. 

\begin{algorithm}[ht]
\DontPrintSemicolon
\caption{\textsc{Fast Deterministic Local Search Algorithm}} \label{alg:fastlocal}
    Add to $S_0$ dummy elements to make it a base (without changing its value).\label{line:to_base}\\
    \For{$i = 1$ \KwTo $k \triangleq \lceil \nicefrac{r}{\eps} \rceil$}
    {
        \If{there exist elements $u\in S_{i - 1}$ and $v\in \cN\setminus S_{i - 1}$ such that $S_{i - 1}-u+v\in \cI$ and $f(v \mid S_{i - 1})- f(u \mid S_{i - 1}-u) \geq 0$\label{line:condition_improving}}{
            Let $u_i$ and $v_i$ be such elements $u$ and $v$ maximizing the difference $f(v_i \mid S_{i - 1})- f(u_i \mid S_{i - 1}-u_i)$, and set $\Delta_i$ to be the value of this difference.\\
			Let $S_i \gets S_{i - 1}+v_i-u_i$.\label{line:change_S}
        }
        \Else
        {
            Let $S_i \gets S_{i - 1}$, and set $\Delta_i \gets 0$.\label{line:keep_S}
        }
	}
    Choose $i^* \in \arg \min_{i \in [k]} \Delta_i$.\\
    \Return $S_{i^* - 1}$.
\end{algorithm}

We are now ready to prove that Algorithm~\ref{alg:fastlocal} has the properties guaranteed by Proposition~\ref{prop:fast}.

\begin{lemma} \label{lem:deterministic}
Algorithm~\ref{alg:fastlocal} can be implemented using $O(\eps^{-1}nr)$ value oracle and $O(\eps^{-1}nr \log r)$ independence oracle queries. Moreover, the final subset $S$ in the algorithm satisfies Inequality~\eqref{ineq-localOPT}.
\end{lemma}
\begin{proof}
For every $i \in [k]$, $\Delta_i$ lower bounds the difference $f(S_i) - f(S_{i - 1})$. If $S_i$ is set by Line~\ref{line:keep_S} of Algorithm~\ref{alg:fastlocal}, then this is trivial since $\Delta_i = 0 = f(S_i) - f(S_{i - 1})$ in this case. Otherwise, by the submodularity of $f$, we have
\[
    \Delta_i
    =
    f(v_i \mid S_{i - 1}) - f(u_i \mid S_{i - 1} - u_i)
    \leq
    f(v_i \mid S_{i - 1} - u_i) - f(u_i \mid S_{i - 1} - u_i)
    =
    f(S_i) - f(S_{i - 1})
    \enspace. 
\]
Observe also that $f(S_k) \leq f(OPT)$ because $S_k$ is a feasible solution. Thus,
\[
    \sum_{i = 1}^k \Delta_i
    \leq
    \sum_{i = 1}^k [f(S_i) - f(S_{i - 1})]
    =
    f(S_k) - f(S_0)
    \leq
    f(OPT) - f(S_0)
    \enspace.
\]
Since $\Delta_{i^*}$ is the minimal value in the above sum, the last inequality implies that
\[
    \Delta_{i^*}
    =
    \min_{i \in [k]} \Delta_i
    \leq
    \frac{1}{k} \sum_{i = 1}^k \Delta_i
    \leq
    \frac{f(OPT) - f(S_0)}{k}
    \enspace.
\]

Using the last observation, let us prove that the output set $S_{i^* - 1}$ of Algorithm~\ref{alg:fastlocal} satisfies Inequality~\eqref{ineq-localOPT} with respect to every independent set $T \in \cI$. Note that $S_{i^* - 1}$ is a base of $\cM$ since it has the size of a base. Additionally, we assume below that $T$ is a base. This assumption is without loss of generality since if $S_{i^* - 1}$ violates Inequality~\eqref{ineq-localOPT} with respect to some set $T$, then it violates it also with respect to any base of $\cM$ obtained by adding enough dummy elements to $T$. 
Therefore, we can use Corollary~\ref{cor:perfect_matching_two_bases} to get a bijective function $h\colon T \to S_{i^* - 1}$ such that $S_{i^* - 1} - h(v) + v \in \cI$ for every $v \in T$ and $h(v) = v$ for every $v \in S_{i^* - 1} \cap T$. We claim that for every element $v \in T$, it holds that
\begin{equation} \label{eq:element_approx_local}
	f(v \mid S_{i^* - 1} - v)- f(h(v) \mid S_{i^* - 1}-h(v))  \leq \Delta_{i^*}
    \enspace.
\end{equation}
If $v$ belongs to $S_{i^* - 1} \cap T$, then $h(v) = v$, and thus, $f(v \mid S_{i^* - 1} - v)- f(h(v) \mid S_{i^* - 1}-h(v)) = 0 \leq \Delta_{i^*}$ because $\Delta_i \geq 0$ for every $i \in [k]$. For $v \in T \setminus S_{i^* - 1}$, we prove Inequality~\eqref{eq:element_approx_local} by proving that
\[
    f(v \mid S_{i^* - 1} - v)- f(h(v) \mid S_{i^* - 1}-h(v))
    = f(v \mid S_{i^* - 1})- f(h(v) \mid S_{i^* - 1}-h(v))
    \leq
    \Delta_{i^*}
    \enspace.
\]
If Algorithm~\ref{alg:fastlocal} sets $S_{i^*}$ on Line~\ref{line:change_S}, then this inequality holds since $h(v)$ and $v$ are one possible pair of candidates that Algorithm~\ref{alg:fastlocal} could choose as $u_{i^*}$ and $v_{i^*}$, respectively. Otherwise, if Algorithm~\ref{alg:fastlocal} sets $S_{i^*}$ on Line~\ref{line:keep_S}, then the last inequality still holds since in this case $\Delta_{i^*} = 0 > f(v \mid S_{i^* - 1})- f(h(v) \mid S_{i^* - 1}-h(v))$ because $h(v)$ and $v$ were one possible pair of candidates to be the elements $u$ and $v$ in the condition on Line~\ref{line:condition_improving} of the algorithm.
Summing up Inequality~\eqref{eq:element_approx_local} over all $v\in T$ and using that $k = \lceil \nicefrac{r}{\eps} \rceil$ now gives
\[\sum_{v\in T}f(v \mid S_{i^* - 1} - v)- \sum_{u\in S_{i^* - 1}}f(u \mid S_{i^* - 1}-u) \leq r \Delta_{i^*} \leq \frac{r \cdot [f(OPT) - f(S_0)]}{k} \leq \eps \cdot [f(OPT) - f(S_0)] \enspace.\]

Next, we would like to determine the number of oracle queries used by Algorithm~\ref{alg:fastlocal}. Line~\ref{line:to_base} of
Algorithm~\ref{alg:fastlocal} requires no oracle queries since it can be implemented by simply adding $r - |S_0|$ dummy elements to $S_0$. In the following, we show that each iteration of Algorithm~\ref{alg:fastlocal} can be implemented using $O(n)$ value oracle queries and $O(n \log r)$ independence oracle queries. Since Algorithm~\ref{alg:fastlocal} has only $O(\frac{r}{\eps})$ iterations, this will imply that Algorithm~\ref{alg:fastlocal} requires only $O(\frac{r}{\eps} \cdot n) = O(\eps^{-1}nr)$ value oracle queries and $O(\frac{r}{\eps} \cdot n \log r) = O(\eps^{-1}nr \log r)$ independence oracle queries.


To implement iteration number $i$ of Algorithm~\ref{alg:fastlocal}, we first compute for each element $u\in S_{i - 1}$ the value $f(u \mid S_{i - 1}-u)$, which requires $O(r) = O(n)$ value oracle queries. By defining the weight of every $u \in S_{i-1}$ to be $f(u \mid S_{i - 1} - u)$, we can then use the algorithm from Lemma~\ref{lem:independence} to find for each $v\in \cN \setminus S_{i - 1}$, an element $u_v \in \arg \min_{u\in S_{i - 1}, S_{i-1}+v-u\in\cI}\{f(u \mid S_{i-1}-u)\}$ using $O(\log r)$ independence oracle queries per element $v$, and thus, $O(n \log r)$ independence oracle queries in total.\footnote{Technically, to use Lemma~\ref{lem:independence} we need to verify, for every element $v \in \cN \setminus S_{i-1}$, that $S_{i-1} + v$ is not independent, and that $S_{i-1} \setminus S_{i - 1} + v = \{v\} \in \cI$. The first of these properties is guaranteed to hold since $S_{i-1}$ is a base of $\cM$. The second property can be checked with a single additional independence oracle query per element $v$. If this property is violated, then by the down-closedness property of independence systems, there is no element $u \in S_{i-1}$ such that $S_{i-1} + v - u \in \cI$.} Once this is done, to complete the iteration, it only remains to find the element $v \in \cN \setminus S_{i-1}$ for which the difference $f(v \mid S_{i - 1}) - f(u_v \mid S_{i - 1} - u_v)$ is maximized, which can be done using only $O(n)$ additional value oracles queries by simply trying all possible options for $v$. If this difference turns out to be non-negative, then we know that the condition in Line~\ref{line:condition_improving} of Algorithm~\ref{alg:fastlocal} evaluates to TRUE in iteration $i$, and the elements $u_i$ and $v_i$ can be set to $u_v$ and $v$, respectively. Otherwise, the condition in Line~\ref{line:condition_improving} evaluates to FALSE in this iteration.
\end{proof}

\subsection{A Randomized Algorithm} \label{ssc:randomized}

In this section, we prove the part of Proposition~\ref{prop:fast} related to its randomized algorithm. We do this by proving the following proposition.

\begin{proposition} \label{prop:fast_rand}
Let $f\colon 2^\cN \rightarrow \bR$ be a submodular function, $M=(\cN,\cI)$ be a matroid of rank $r$ over a ground set $\cN$ of size $n$, $S_0$ be an arbitrary given set in $\cI$, and $OPT$ be a set in $\cI$ maximizing $f$. Then, for any $\eps\in (0,1)$, there exists a randomized algorithm that makes $\tilde{O}(\eps^{-1}(n+r\sqrt{n}))$ value and independence oracle queries and outputs a subset $S\in \cI$ and a value $\Delta$ that obey
\begin{equation} \label{eq:delta_guarantee}
    \sum_{v \in T} f(v \mid S - v) - \sum_{u \in S} f(u \mid S - u) \leq \Delta
    \qquad
    \forall T \in \cI
    \enspace.
\end{equation}
Furthermore, with probability at least $1/2$, $\Delta \leq \frac{\eps}{3} \cdot [f(OPT) - f(S_0)]$.
\end{proposition}

The guarantee of Proposition~\ref{prop:fast} for its randomized algorithm follows from Proposition~\ref{prop:fast_rand} by a repetition argument. Specially, to get the guarantee of Proposition~\ref{prop:fast}, one needs to execute the algorithm of Proposition~\ref{prop:fast_rand} $p \triangleq \lceil \log_{2}{\eps^{-1}} \rceil$ times, which results in a list of $p$ pairs $(S_1, \Delta_1), (S_2, \Delta_2), \dotsc, (S_p, \Delta_p)$ (and increases the query complexity just by a logarithmic factor). Let $\Delta_i$ be the lowest $\Delta$ value in this list. The set $S_i$ is the output set described by Proposition~\ref{prop:fast}. Notice that with probability at least $1 - (1 - 1/2)^p \geq 1 - \eps$, we have $\Delta_i \leq \frac{\eps}{3} \cdot [f(OPT) - f(S_0)] \leq \eps \cdot [f(OPT) - f(S_0)]$ (the last inequality holds since the independence of $S_0$ in the matroid implies that $f(OPT) \geq f(S_0)$), and thus, $S_i$ obeys Inequality~\eqref{ineq-localOPT} with probability at least $1 - \eps$. If $f$ is guaranteed to be monotone and $S_0 = \varnothing$, one can add an additional step to the algorithm in which it first applies the algorithm from Lemma~\ref{lem:fastapprox} to the function $f'(S) = f(S) - f(\varnothing)$ (which is non-negative, monotone and submodular) to get an independent set $S'$ obeying $f(S') - f(\varnothing) \geq \frac{1}{3} \cdot [f(OPT) - f(\varnothing)]$, and then indicates failure if $\Delta_i > \eps \cdot [f(S') - f(S_0)]$. This step indicates failure with probability at most
\begin{multline*}
    \Pr[\Delta_i > \eps \cdot [f(S') - f(S_0)]]
    \leq
    \Pr\Big[\Delta_i > \eps \cdot \Big[\frac{1}{3}\cdot f(OPT) + \frac{2}{3} \cdot f(\varnothing) - f(S_0)\Big]\Big]\\
    =
    \Pr\Big[\Delta_i > \frac{\eps}{3} \cdot [f(OPT) - f(S_0)]\Big]
    =
    1 - \Pr\Big[\Delta_i \leq \frac{\eps}{3} \cdot [f(OPT) - f(S_0)]\Big]
    \leq
    \eps
    \enspace.
\end{multline*}
Furthermore, we claim that the output set $S_i$ obeys Inequality~\eqref{ineq-localOPT} whenever this step does not indicate failure. This is the case because the fact that $S'$ is a feasible solution implies that, whenever a failure is not indicated, it must hold that $\Delta_i \leq \eps \cdot [f(S') - f(S_0)] \leq \eps \cdot [f(OPT) - f(S_0)]$.

The algorithm we use to prove Proposition~\ref{prop:fast_rand} appears as Algorithm~\ref{alg:fastlocalrandom}. For simplicity, we assume in this algorithm that $n$ is a perfect square. This can be guaranteed, for example, by adding additional dummy elements to the ground set. Alternatively, one can simply replace every appearance of $\sqrt{n}$ in Algorithm~\ref{alg:fastlocalrandom} with $\lceil \sqrt{n} \rceil$. Some of the ideas in the design and analysis of Algorithm~\ref{alg:fastlocalrandom} can be traced to the work of Tukan et al.~\cite{tukan2024practical}, who designed a fast local search algorithm for maximizing a submodular function subject to a cardinality constraint.


\begin{algorithm}[ht]
\caption{\textsc{Fast Randomized Local Search Algorithm}} \label{alg:fastlocalrandom}
\DontPrintSemicolon
    Add to $S_0$ dummy elements to make it a base (without changing its value).\label{line:to_base_random1}\\
    \For{$i=1$ \KwTo $k\triangleq \lceil \nicefrac{12r}{\eps} \rceil$}{
			Sample a subset $R_1\subseteq S_{i-1}$ of size $\min\{r, \sqrt{n}\}$ uniformly at random.\\
			Sample a subset $R_2\subseteq \cN$ of size $\max\{\frac{n}{r}, \sqrt{n}\}$ uniformly at random.\\
            Let $R'_2 \triangleq \{v \in R_2 \setminus S_{i - 1} \mid S_{i - 1} \setminus R_1 + v \in \cI\}$.\\
		Initialize $S_i\gets S_{i-1}$.\\
            \If{$R'_2 \neq \varnothing$}
            {
             Compute $f(v\mid S_{i-1})$ for every $v\in R'_2$ and $f(u\mid S_{i-1}-u)$ for every $u\in R_1$.\label{line:compute}\\
              \For{{\bf each} $v\in R'_2$\label{line:loop11}
              }{$u_{v}\gets \arg \max_{u\in R_1, S_{i-1}-u+v\in \cI}\{f(v \mid S_{i-1})-f(u\mid S_{i-1}-u)\}$.\label{line:computeuv}
              }
              $v^*\gets \arg\max_{v\in R'_2} \{f(v \mid S_{i-1})-f(u_v\mid S_{i-1}-u_v)\}$. \label{line-v1}\\
             \If{$f(v^* \mid S_{i-1})-f(u_{v^*}\mid S_{i-1}-u_{v^*})\geq 0$}{Update $S_{i} \gets S_{i-1}+v^*-u_{v^*}$. \label{line-v2}}
             }
}
Sample uniformly at random $i\in [k]$.\\
Let $\Delta \gets \max_{T \in \cI} [\sum_{v\in T}f(v \mid S_{i - 1} - v) -\sum_{u\in S_{i - 1}}f(u \mid S_{i - 1}-u)]$.\label{line:find_delta}\\
\Return $S_i$ and $\Delta$.
\end{algorithm}


The next two lemmata show together that Algorithm~\ref{alg:fastlocalrandom} obeys all the properties guaranteed by Proposition~\ref{prop:fast_rand}. The first of these lemmata bounds the number of queries used by Algorithm~\ref{alg:fastlocalrandom}, and the second lemma bounds the probability that $\Delta$ is larger than $\frac{\eps}{3} \cdot f(OPT)$ (notice that $\Delta$ obeys Inequality~\eqref{eq:delta_guarantee} by definition). 

\begin{lemma}
  Algorithm \ref{alg:fastlocalrandom} makes only $O(\eps^{-1}(n+r\sqrt{n}))$ value oracle queries and $O(\eps^{-1}(n+r\sqrt{n})\log r)$ independence oracle queries.
\end{lemma}

\begin{proof}
The first line of Algorithm~\ref{alg:fastlocalrandom} requires no oracle queries since it can be implemented by adding $r - |S_0|$ dummy elements to $S$. 
Line~\ref{line:find_delta} of the algorithm boils down to finding an independent set $T$ maximizing a linear function $g\colon 2^\cN \to \nnR$ defined as $g(T) \triangleq \sum_{v \in T} f(v \mid S_i - v)$ for every $T \subseteq \cN$. As explained in Section~\ref{sec:preliminaries}, such a set can be found using $O(n)$ value and independence queries via a well-known greedy algorithm, and thus, Line~\ref{line:find_delta} can be implemented using this number of queries. Given these observations, to complete the proof of the lemma, it only remains to show that every iteration of the main loop of Algorithm~\ref{alg:fastlocalrandom} can be implemented using $O(n/r + \sqrt{n})$ value oracle queires and $O((n/r + \sqrt{n}) \log r)$ independence oracle queries.

Value oracle queries are necessary only for implementing Line~\ref{line:compute} in this loop, and the algorithm needs $O(|R_1|+|R'_2|) = O(|R_1|+|R_2|)=O(n/r + \sqrt{n})$ such queries per iteration to implement this line. Independence oracle queries are necessary for implementing two things in the main loop of Algorithm~\ref{alg:fastlocalrandom}: the computation of $R'_2$ and the loop starting on Line~\ref{line:loop11}. Computing the set $R'_2$ requires $O(|R_2|) = O(n/r + \sqrt{n})$ independence oracle calls. Every iteration of the loop starting on Line~\ref{line:loop11} can be implemented by invoking Lemma~\ref{lem:independence} with $S'=R_1$ and $w_u=f(u \mid S_{i-1}-u)$, which requires $O(\log r)$ independence oracle queries. Hence, overall this loop requires $O(|R'_2|\log r) = O(|R_2|\log r)= O((n/r+\sqrt{n})\log r)$ independence oracle queries (per iteration of the main loop).
\end{proof}

\begin{lemma}\label{lem:fail}
   With probability at least $1/2$, the value of $\Delta$ returned by Algorithm~\ref{alg:fastlocalrandom} is no more than $\frac{\eps}{3} \cdot [f(OPT) - f(S_0)]$. 
\end{lemma}

\begin{proof}
Recall that Algorithm~\ref{alg:fastlocalrandom} makes $k=\lceil \frac{12r}{\eps} \rceil$ iterations, and let $i\in[k]$ be one of these iterations. Fix all randomness of the algorithm until iteration $i$ (including iteration $i$ itself), 
and let $T$ be any base of $\cM$. Since $S_{i - 1}$ is also a base of $\cM$ (it is an independent set with the size of a base), Corollary~\ref{cor:perfect_matching_two_bases} guarantees the existence of a bijection $h\colon T \to S_{i-1}$ such that $S_{i-1} - h(v) + v \in \cI$ for every $v \in T$ and $h(v) = v$ for every $v \in S_{i-1} \cap T$. Let us now define $C\triangleq \{(h(v),v) \mid v\in T\}$, and let $R_1$, $R_2$, $R'_2$, $v^*$ and $u_{v^*}$ denote their values in iteration number $i$. If $R'_2 \neq \varnothing$, then
{\allowdisplaybreaks\begin{align} \label{eq:deterministic_inequality}
  f(S_{i} \mid S_{i-1})
	& =
	\max\{0, f(S_{i-1}-u_{v^*}+v^*)- f(S_{i-1})]\}\\\nonumber
    & \geq
	\max\{0, f(v^* \mid S_{i-1})- f(u_{v^*} \mid S_{i-1}-u_{v^*})]\}\\
  &=
	\max\{0,\max_{\substack{(u,v)\in R_1\times R_2'\\S_{i-1}-u+v\in \cI}}[f(v \mid S_{i-1})- f(u \mid S_{i-1}-u)]\}\\\nonumber
  & \geq
	\max\{0,\max_{(u,v)\in C \cap (R_1\times (R_2 \setminus S_{i - 1}))}[f(v \mid S_{i-1})- f(u \mid S_{i-1}-u)]\}\\
	& =
	\max\{0,\max_{(u,v)\in C \cap (R_1\times R_2)}[f(v \mid S_{i-1} - v)- f(u \mid S_{i-1}-u)]\}\\\nonumber
  &
	\geq \max\bigg\{0,\sum_{(u,v)\in C \cap (R_1\times R_2)}\frac{f(v \mid S_{i-1} - v)- f(u \mid S_{i-1}-u)}{|C \cap (R_1\times R_2)|}\bigg\}
\end{align}}%
where the first inequality holds since the facts that $u_{v^*} \in S_{i - 1}$ and $v^* \not \in S_{i - 1}$ and the submodularily of $f$ imply together that $f(u_{v^*} \mid S_{i-1}-u_{v^*})\geq f(u_{v^*} \mid S_{i-1}-u_{v^*}+v^*)$. The second equality follows as $(u_{v^*}, v^*)$ is the pair maximizing $\max_{(u,v)\in R_1\times R'_2,S_{i-1}-u+v\in \cI}[f(v \mid S_{i-1})- f(u \mid S_{i-1}-u)]$.
The second inequality holds since $(u, v) \in C \cap (R_1 \times (R_2 \setminus S_{i - 1}))$ implies $S_{i - 1} \setminus R_1 + v \subseteq S_{i - 1} - u + v \in \cI$, and the third equality holds since for every $(u, v) \in C$ such that $v \in S_{i - 1}$ it holds that $v = h(v) = u$. 
If $R'_2 = \varnothing$, then $S_i = S_{i - 1}$ and $C \cap (R_1 \times R_2) = \varnothing$, and therefore, the leftmost and rightmost sides of Inequality~\eqref{eq:deterministic_inequality} are both $0$, which means that the leftmost side still upper bounds the rightmost side in this case.

Unfixing the random choice of the subsets $R_1$ and $R_2$ in iteration $i$, we get

\begin{align}
  \bE_{R_1, R_2}[f(S_{i} \mid S_{i-1})] 
  & \geq \bE\bigg[\max\bigg\{0,\sum_{(u,v)\in C \cap (R_1\times R_2)}\frac{f(v \mid S_{i-1} - v)- f(u \mid S_{i-1}-u)}{|C \cap (R_1\times R_2)|}\bigg\}\bigg] \label{ineq-to-bound}\\\nonumber  & \geq \frac{1-\nicefrac{1}{e}}{r}\bigg[\sum_{v\in T}f(v\mid S_{i-1} - v) - \sum_{u\in S_{i-1}}f(u \mid S_{i-1}-u)\bigg] \enspace.
\end{align}
Let us explain why the last inequality holds.
Let $p_{\geq 1}$ be the probability that $C \cap (R_1\times R_2)\neq \varnothing$. Assume now that we select a pair $(u,v)\in C$ uniformly at random among the pairs in $C \cap (R_1\times R_2)$, and let $p_{u,v}$ be the probability that the pair $(u, v) \in C$ is selected conditioned on $C \cap (R_1\times R_2)$ being non-empty.
Then, for each pair $(u,v)\in C$ the probability it is chosen is exactly the probability that $C \cap (R_1\times R_2)\neq \varnothing$ times the probability it is chosen given this event. We therefore get that,
\begin{multline*}
	\bE\bigg[\sum_{(u,v)\in C \cap (R_1\times R_2)} \mspace{-18mu}\frac{f(v \mid S_{i-1} - v)- f(u \mid S_{i-1}-u)}{|C \cap (R_1\times R_2)|}\bigg]\\
	=
	p_{\geq 1}\cdot \sum_{(u,v)\in C}p_{u,v}\left[f(v\mid S_{i-1} - v)-f(u \mid S_{i-1}-u)\right] \enspace.
\end{multline*}
As the choice of the sets $R_1$ and $R_2$ is done uniformly at random, the probability $p_{u, v}$ must be the same for every $(u,v)\in C$, and thus, it is $1/|C| = 1/r$. Note also that the probability $p_{\geq 1}$ that $C \cap (R_1\times R_2)$ is non-empty is at least $1 - \nicefrac{1}{e}$ because, even conditioned on a fixed choice of the set $R_1\subseteq S_{i-1}$, the probability that none of the elements in $\{v \mid h(v) \in R_1\}$ are chosen to $R_2$ is only

\[\prod_{i=1}^{|R_2|}\left(1-\frac{|R_1|}{n-i+1}\right) \leq \left(1-\frac{|R_1|}{n}\right)^{|R_2|} \leq e^{-\nicefrac{|R_1|\cdot |R_2|}{n}}= \frac{1}{e} \enspace, \]
where the final equality holds since  $|R_1|\cdot |R_2|= \min\{r, \sqrt{n}\}\cdot \max\{\frac{n}{r}, \sqrt{n}\}=n$. This completes the proof of Inequality~\eqref{ineq-to-bound}.



 
Next, suppose that there exists a subset $T\in \cI$ such that
\begin{equation}\label{eq:fail_inequality}\sum_{v\in T}f(v \mid S_{i-1} - v)-\sum_{u\in S_{i-1}}f(u \mid S_{i-1}-u) \geq \frac{\eps}{3} \cdot [f(OPT) - f(S_0)] \enspace. \end{equation}
We may assume that $T$ is a base of $\cM$ because if $T$ obeys Inequality~\eqref{eq:fail_inequality}, then every base obtained by adding enough dummy elements to $T$ also obeys this inequality. Thus, by Inequality~\eqref{ineq-to-bound}, $\bE_{R_1, R_2}[f(S_{i} \mid S_{i-1})]\geq \frac{1-\nicefrac{1}{e}}{r}\cdot \frac{\eps}{3} \cdot [f(OPT) - f(S_0)] \geq \frac{\eps}{6r} \cdot [f(OPT) - f(S_0)]$.

Unfix now the randomness of all the iterations of Algoritm~\ref{alg:fastlocalrandom}, and let us define $\cA_i$ to be the event that in iteration $i$ there exists a subset $T\in \cI$ obeying Inequality~\eqref{eq:fail_inequality}. By the above inequality and the law of total expectation,
\begin{align*}
  \bE[f(S_{i} \mid S_{i-1})]
	={} &
	\Pr[\cA_i] \cdot \bE[f(S_{i} \mid S_{i-1}) \mid \cA_i] + \Pr[\bar{\cA}_i] \cdot \bE[f(S_{i} \mid S_{i-1}) \mid \bar{\cA}_i]\\
	\geq{} &
	\Pr[\cA_i] \cdot \bE[f(S_{i} \mid S_{i-1}) \mid \cA_i]
	\geq
	\Pr[\cA_i] \cdot \frac{\eps}{6r} \cdot [f(OPT) - f(S_0)]
	\enspace,
\end{align*}
where the first inequality uses the fact that the inequality $f(S_i)\geq f(S_{i-1})$ holds deterministically. Summing up the last inequality over all $i \in [k]$ yields
\[
  f(OPT) - f(S_0) \geq \bE[f(S_{k})] - f(S_0) = \sum_{i=1}^{k} \bE[f(S_{i} \mid S_{i-1})] \geq \frac{\eps}{6r} \cdot \sum_{i=1}^{k} \Pr[\cA_{i}] \cdot [f(OPT) - f(S_0)]\enspace,
\]
where the first inequality holds because $S_k$ is an independent set of $\cM$.
Hence, $\sum_{i=1}^{k} \Pr[\cA_{i}] \leq 6r/\eps$. It now remains to observe that, by the definition of $\cA_i$, $\Pr[\cA_i]$ is exactly the probability that the algorithm returns a value $\Delta > \frac{\eps}{3} \cdot [f(OPT) - f(S_0)]$ if it chooses $S_{i-1}$ as its output set. Since the algorithm chooses a uniformly random $i \in [k]$ for this purpose, the probability that it returns such a value $\Delta$ is $k^{-1} \cdot \sum_{i=1}^{k} \Pr[\cA_{i}] \leq \frac{6r}{k \cdot\eps} \leq \frac{1}{2}$.
\end{proof}

\section{Conclusion}

In this work, we designed a deterministic non-oblivious local search algorithm that has an approximation guarantee of $1 - \nicefrac{1}{e} - \eps$ for the problem of maximizing a monotone submodular function subject to a matroid constraint. This algorithm shows that there is essentially no separation between the approximation guarantees that can be obtained by deterministic and randomized algorithms for maximizing monotone submodular functions subject to a general matroid constraint.
Adding randomization, we showed that the query complexity of our algorithm can be improved to $O_\eps(n + r\sqrt{n})$. Following this work, Buchbinder \& Feldman~\cite{buchbinder2024extending} developed a greedy-based \emph{deterministic} algorithm for maximizing general \emph{non-monotone} submodular functions subject to a matroid constraint, which, for the special case of monotone submodular functions, recovers the optimal approximation guarantee of Theorem~\ref{thm:main}. Like our algorithm, the algorithm of~\cite{buchbinder2024extending} also has both deterministic and randomized versions. Its deterministic version has a much larger query complexity compared to our deterministic algorithm, but the randomized version of the algorithm of~\cite{buchbinder2024extending} has a query complexity of $\tilde{O}_\eps(n + r^{3/2})$, which is always at least as good as the query complexity our randomized algorithm.

There are plenty of open questions that remain. The most immediate of these is to find a way to further improve the running time of the randomized (or the deterministic) algorithm to be nearly-linear for all values of $r$ (rather than just for $r = \tilde{O}(n^{2/3})$ as in~\cite{buchbinder2024extending}). A more subtle query complexity question is related to the fact that, since a value oracle query for the function $g'$ defined in Section~\ref{ssc:implementation} requires $O(2^{\nicefrac{1}{\eps}})$ value oracle queries to $f$, the query complexities of our algorithms depend exponentially on $1/\eps$. We remark that the algorithms designed in~\cite{buchbinder2024extending} also have exponential dependency on $1/\eps$.
In other words, our algorithms and the algorithms of~\cite{buchbinder2024extending} are efficient PTASs in the sense that their query complexities are $h(\eps) \cdot O(\Poly(n, r))$ for some function $h$. A natural question is whether it is possible to design a deterministic FPTAS for the problem we consider, i.e., a deterministic algorithm obtaining $(1 - 1/e - \eps)$-approximation using $O(\Poly(n, r, \eps^{-1}))$ oracle queries. By guessing the most valuable element of an optimal solution and setting $\eps$ to be polynomially small, such an FPTAS will imply, in particular, a deterministic algorithm guaranteeing a clean approximation ratio of $1 - 1/e$ using $O(\Poly(n, r))$ oracle queries.


\appendix
\section{Sums of Monotone Submodular and Linear Functions} \label{app:submodular_linear_sums}

\newcommand{\linearFunction}{{b}}

Sviridenko et al.~\cite{sviridenko2017optimal} initiated the study of the following problem. Given a non-negative monotone submodular function $f \colon 2^\cN \to \nnR$, a linear function $\linearFunction \colon 2^\cN \to \bR$ and a matroid $\cM = (\cN, \cI)$, find a set $S \in \cI$ that (approximately) maximizes $f(S) + \linearFunction(S)$.\footnote{It is customary in the literature to denote the linear function by $\ell$ in this context. However, in this paper, we use instead the notation $\linearFunction$ to avoid confusion with the parameter $\ell$ of our algorithms.} They showed that one can efficiently find a set $S \in \cI$ such that
\begin{equation} \label{eq:target_guarantee}
	f(S) + \linearFunction(S)
	\geq
	(1 - 1/e) \cdot f(OPT) + \linearFunction(OPT) - \{\text{small error term}\}
	\enspace,
\end{equation}
and that this is the best that can be done (in some sense). The original motivation of Sviridenko et al.~\cite{sviridenko2017optimal} for studying this problem has been obtaining optimal approximation for maximization of monotone submodular functions of bounded curvature subject to a matroid constraint.\footnote{The curvature $c$ is a number between $[0, 1]$ that measures the distance of a monotone submodular function from linearity. Sviridenko et al.~\cite{sviridenko2017optimal} obtained an approximation ratio of $1 - c/e$ for this problem, which smoothly interpolates between the optimal approximation ratio of $1 - 1/e$ for general monotone submodular functions (the case of $c \leq 1$) and the optimal approximation ratio of $1$ for linear functions (the case of $c = 0$). Sviridenko et al.~\cite{sviridenko2017optimal} also showed that their result is optimal for any value of $c$.} Later this problem found uses also in machine learning applications as it captures soft constraints and regularizers~\cite{harshaw2019submodular,kazemi2021regularized,nikolakaki2021efficient}. Feldman~\cite{feldman2021guess} suggested a simpler algorithm for maximizing sums of monotone submodular functions and linear functions subject to arbitrary solvable polytope constraints, and maximization of such sums involving non-monotone submodular functions was also studied~\cite{bodek2022maximizing,qi2024maximizing,sun2023regularized}.

Recall now that Filmus and Ward~\cite{filmus2014monotone} described a non-oblivious local search algorithm guided by some auxiliary objective function $g$ such that the output set $S$ of this local search is guaranteed to obey $f(S) \geq (1 - 1/e) \cdot f(OPT)$. Furthermore, by applying a local search algorithm to the linear function $\linearFunction$, it is possible to get an output set $S$ such that $\linearFunction(S) = \linearFunction(OPT)$. One of the results of Sviridenko et al.~\cite{sviridenko2017optimal} was a non-oblivious local search based on the guiding function $(1 - 1/e) \cdot g(\cdot) + \linearFunction(\cdot)$, which is a weighted linear combination of the two above guiding functions. Sviridenko et al.~\cite{sviridenko2017optimal} showed that the output set $S$ of their local search algorithm obeys Inequality~\eqref{eq:target_guarantee}, which is natural since Inequality~\eqref{eq:target_guarantee} is a linear combination of the above-mentioned guarantees of the local searches guided by either $g$ or $\linearFunction$.

Reusing the idea of Sviridenko et al.~\cite{sviridenko2017optimal}, we can maximize the sum $f(\cdot) + \linearFunction(\cdot)$ by making three changes to Algorithm~\ref{alg:final}. The first change is a replacement of the function $g'$ in Line~\ref{line:maximizing_g_prime} of Algorithm~\ref{alg:final} with $g'(\cdot) + \alpha_\ell \cdot (\ell + 1) \cdot \linearFunction'(\cdot)$, where $\linearFunction'$ is the natural extension of $\linearFunction$ to the ground set $\cN'$ defined, for every set $S \subseteq \cN'$, by $\linearFunction'(S) \triangleq \sum_{(u, j) \in S} \linearFunction(\{u\})$. Notice that $g'(\cdot) + \alpha_\ell \cdot (\ell + 1) \cdot \linearFunction'(\cdot)$ is a (not necessarily monotone or non-negative) submodular function, and therefore, Algorithm~\ref{alg:final} can invoke Proposition~\ref{prop:fast} on this function. The second change we need to make in Algorithm~\ref{alg:final} is that the set $S_0$ passed to the algorithm of Proposition~\ref{prop:fast} should be a set maximizing the linear function $\linearFunction'$ among all independent sets of $\cM'$ (instead of the empty set). As explained in Section~\ref{sec:preliminaries}, such a set can be found via a well-known
greedy algorithm using $O(n)$ oracle queries. The final change we need to make in Algorithm~\ref{alg:final} is that the algorithm should output the better solution among its original output set $\pi_{[\ell]}(S)$ and $\pi_{[\ell]}(S_0)$.

The modified Algorithm~\ref{alg:final} resulting from the above changes obeys the properties stated in the following theorem (for appropriate choice of the parameter $\ell$). To avoid repeating lengthy calculations and arguments, we do not give a full proof of this theorem. Instead, we only explain below how such a proof would differ from the analysis of the original Algorithm~\ref{alg:final}.

\begin{theorem} \label{thm:sum}
There exists a deterministic non-oblivious local search algorithm that given a non-negative monotone submodular function $f \colon 2^\cN \to \nnR$, a linear function $\linearFunction \colon 2^\cN \to \bR$, a matroid $\cM = (\cN, \cI)$ and a value $\eps > 0$, finds a set $S \in \cI$ such that $\bE[f(S) + \linearFunction(S)] \geq (1 - \nicefrac{1}{e}) \cdot f(OPT) + \linearFunction(OPT) - O(\eps) \cdot M$ for every set $OPT \in \cI$, where $M \triangleq \max_{T \in \cI} f(T)$. This algorithm has a query complexity of $\tilde{O}_\eps(nr)$, where $r$ is the rank of the matroid and $n$ is the size of its ground set. The query complexity can be improved to $\tilde{O}_\eps(n + r\sqrt{n})$ using randomization.
\end{theorem}

The modified Algorithm~\ref{alg:final} applies Proposition~\ref{prop:fast} to the function $g'(\cdot) + \alpha_\ell \cdot (\ell + 1) \cdot \linearFunction'(\cdot)$, which results in a set $S \in \cI'$ that with probability at $1 - \eps'$ obeys
\begin{multline*}
    \sum_{v \in T} [g'(v \mid S - v) + \alpha_\ell \cdot (\ell + 1) \cdot \linearFunction'(v \mid S - v)]
    -
    \sum_{u \in S} [g'(u \mid S - u) + \alpha_\ell \cdot (\ell + 1) \cdot \linearFunction'(u \mid S - u)]\\
    \leq
    \eps' \cdot [g'(A) +  \alpha_\ell \cdot (\ell + 1) \cdot \linearFunction'(A) - g'(S_0) - \alpha_\ell \cdot (\ell + 1) \cdot \linearFunction'(S_0)] \qquad \forall T \in \cI'
    \enspace,
\end{multline*}
where $A$ is a set in $\cI'$ maximizing $g'(\cdot) + \alpha_\ell \cdot (\ell + 1) \cdot \linearFunction'(\cdot)$. Since $g'(S_0)$ is non-negative and $\linearFunction'(S_0) \geq \linearFunction'(A)$ by the definition of $S_0$, the right hand side of the last inequality can be replaced by $\eps' \cdot g'(A)$. The inequality obtained after this replacement can be further simplified using the facts that $\linearFunction'$ is a linear function and $\linearFunction'(T) = \linearFunction(\pi_{[\ell]}(T))$ for every set $T \in \cI'$, which yields
\begin{multline} \label{eq:probabilitic_inequality}
    \sum_{v \in T} g'(v \mid S - v) + \alpha_\ell \cdot (\ell + 1) \cdot \linearFunction(\pi_{[\ell]}(T))
    -
    \sum_{u \in S} g'(u \mid S - u) + \alpha_\ell \cdot (\ell + 1) \cdot \linearFunction(\pi_{[\ell]}(S))\\
    \leq
    \eps' \cdot g'(A) \qquad \forall T \in \cI'
    \enspace.
\end{multline}

We now prove the following variant of Lemma~\ref{lem:main}.

\begin{lemma} \label{lem:main-sum}
Let $\alpha_{\ell + 1} \triangleq 0$. If Inequality~\ref{eq:probabilitic_inequality} holds, then
\begin{multline*}
\sum_{J \subseteq [\ell]} 
\Big[\alpha_{|J|} \cdot |J|\cdot \Big(1+\frac{1}{\ell}\Big) -\alpha_{|J|+1} \cdot (\ell-|J|)\Big] \cdot f(\pi_J(S)) + \alpha_\ell \cdot (\ell + 1) \cdot [\linearFunction(\pi_{[\ell]}(S)) - \linearFunction(OPT)] \\ \geq \sum_{i=1}^{\ell}\binom{\ell-1 }{i-1} \cdot \alpha_{i} \cdot f(OPT) - \eps' \cdot g'(A) \enspace.
\end{multline*}
\end{lemma}
\begin{proof}
Consider the proof of Lemma~\ref{lemma:approximation_final}. Using Inequality~\eqref{eq:probabilitic_inequality}, instead of Inequality~\eqref{eq:final_alg_guarantee} in the derivation of the first display math in this proof yields
\begin{multline*}
    \sum_{u \in OPT} \sum_{J \subseteq [\ell], j \in J} \alpha_{|J|} \cdot f(u \mid \pi_J(S)) + \alpha_\ell \cdot (\ell + 1) \cdot \linearFunction(OPT)
	\leq\\
	\sum_{(u, k) \in S} \sum_{J \subseteq [\ell], k \in J} \mspace{-9mu}\alpha_{|J|} \cdot f(u \mid \pi_J(S) - u) + \alpha_\ell \cdot (\ell + 1) \cdot \linearFunction(\pi_{[\ell]}(S)) + \eps' \cdot g'(A)
    \enspace.
\end{multline*}

Notice that this inequality differs from the inequality given by the first display math in the proof of Lemma~\ref{lemma:approximation_final} only by the addition of the $\linearFunction$ involving terms and the replacement of $OPT'$ with $A$ in the error term $\eps \cdot g'(A)$. Dragging these changes through the following steps in the proof of Lemma~\ref{lemma:approximation_final}, we get the following inequality (instead of the third display math in the proof of Lemma~\ref{lemma:approximation_final}).
\begin{multline*}
    \sum_{J \subseteq [\ell]} 
    \left[\alpha_{|J|} \cdot |J| -\alpha_{|J|+1} \cdot (\ell-|J|)\right] \cdot f(\pi_J(S)) + \alpha_\ell \cdot (\ell + 1) \cdot [\linearFunction(\pi_{[\ell]}(S)) - \linearFunction(OPT)] \\
 \geq \sum_{J \subseteq [\ell]}\frac{|J|}{\ell} \cdot \alpha_{|J|} \cdot \left[f(OPT) -f(\pi_J(S))\right] - \eps' \cdot g'(A)
 \enspace.
\end{multline*}
The lemma now follows by rearranging the last inequality since, as was already argued by the proof of Lemma~\ref{lem:main},
\[\sum_{J \subseteq [\ell]}\frac{|J|}{\ell} \cdot \alpha_{|J|} \cdot f(OPT) = \sum_{i=1}^{\ell}\binom{\ell }{i} \cdot \frac{i}{\ell} \cdot \alpha_{i} \cdot f(OPT) = \sum_{i=1}^{\ell}\binom{\ell-1 }{i-1} \cdot \alpha_{i} \cdot f(OPT)
\enspace.
\qedhere
\]
\end{proof}

\begin{corollary} \label{cor:sum}
If Inequality~\eqref{eq:probabilitic_inequality} holds, then
\[
    f(\pi_{[\ell]}(S)) + \linearFunction(\pi_{[\ell]}(S))
     \geq [1 - (1 + 1/\ell)^{-\ell}] \cdot f(OPT) + (1 + 1/\ell)^{-\ell} \cdot f(\varnothing) + \linearFunction(OPT) - \eps M
     \enspace.
\]
\end{corollary}
\begin{proof}
Using the properties of the coefficients $\alpha_0, \alpha_1, \dotsc, \alpha_\ell$ described in the proof of Proposition~\ref{prop:approximation_basic}, the inequality guaranteed by Lemma~\ref{lem:main-sum} simplifies to
\begin{multline} \label{eq:bound_with_A}
    (1 + 1/\ell)^{\ell - 1} \cdot (\ell + 1) \cdot f(\pi_{[\ell]}(S)) - \ell \cdot f(\varnothing) + (1 + 1/\ell)^{\ell - 1} \cdot (\ell + 1) \cdot [\linearFunction(\pi_{[\ell]}(S)) - \linearFunction(OPT)] \\
 \geq \ell \cdot [(1 + 1/\ell)^\ell - 1] \cdot f(OPT) - \eps' \cdot g'(A)
 \enspace.
\end{multline}
We now recall the bound $\frac{\eps' \cdot g'(OPT')}{\alpha_\ell \cdot (\ell + 1)} \leq \eps \cdot f(OPT)$ from the proof of Lemma~\ref{lemma:approximation_final}. This bound was based on the fact that for every set $J \subseteq [\ell]$ it held in the setting of Section~\ref{sec:algorithm} that $\pi_J(OPT') \subseteq \pi_{[\ell]}(OPT') \in \cI$, and therefore, $f(\pi_J(OPT')) \leq f(OPT)$ by the definition of $OPT$ used in that section. Similarly, in the current section, we have $\pi_J(A) \subseteq \pi_{[\ell]}(A) \in \cI$, and therefore, $f(\pi_J(A)) \leq M$ by the definition of $M$. Thus, repeating the arguments from the proof of Lemma~\ref{lemma:approximation_final}, one can get the analogous bound
\[
	\frac{\eps' \cdot g'(A)}{\alpha_\ell \cdot (\ell + 1)}
	\leq
	\eps M
    \enspace.
\]
Plugging this inequality into Inequality~\eqref{eq:bound_with_A}, and rearranging yields the corollary.
\end{proof}

We are now ready to prove Theorem~\ref{thm:sum}.

\begin{proof}[Proof of Theorem~\ref{thm:sum}.]
We show below that when $\ell$ is set to be $1 + \lceil 1 / \eps \rceil$, the modified Algorithm~\ref{alg:final} has all the properties guaranteed by Theorem~\ref{thm:sum}. First, one can verify that the modifications of Algorithm~\ref{alg:final} do not affect its asymptotic query complexities, and therefore, the algorithm still has the query complexities stated in Theorem~\ref{thm:main}, which are identical to the query complexities guaranteed by Theorem~\ref{thm:sum}. Thus, in the following, we focus on lower bounding the value of the output set $\hat{S}$ of the modified Algorithm~\ref{alg:final}.

Let $\cE$ be the event that the set $S$ produced by the modified Algorithm~\ref{alg:final} obeys Inequality~\eqref{eq:probabilitic_inequality}. By the discussion before the introduction of Inequality~\eqref{eq:probabilitic_inequality}, it holds that $\Pr[\cE] \geq 1 - \eps' \geq 1 - \eps$. Thus, by the law of total expectation,
\begin{align*}
    \bE[f(\hat{S}) + \linearFunction(\hat{S})]
    ={} &
    \Pr[\cE] \cdot \bE[f(\hat{S}) + \linearFunction(\hat{S}) \mid \cE] + \Pr[\bar{\cE}] \cdot \bE[f(\hat{S}) + \linearFunction(\hat{S}) \mid \bar{\cE}]\\
    \geq{} &
    \Pr[\cE] \cdot \bE[f(\pi_{[\ell]}(S)) + \linearFunction(\pi_{[\ell]}(S)) \mid \cE] + (1 - \Pr[\cE]) \cdot [f(S_0) + \linearFunction(S_0)]\\
    \geq{} &
    \Pr[\cE] \cdot \{[1 - (1 + 1/\ell)^{-\ell}] \cdot f(OPT) + (1 + 1/\ell)^{-\ell} \cdot f(\varnothing) + \linearFunction(OPT) - \eps M\} \\&\mspace{100mu}+ (1 - \Pr[\cE]) \cdot [f(S_0) + \linearFunction(S_0)]\\
    \geq{} &
    \Pr[\cE] \cdot [1 - (1 + 1/\ell)^{-\ell}] \cdot f(OPT) + \linearFunction(OPT) - \eps M\\
    \geq{} &
    (1 - \eps) \cdot [1 - (1 + 1/\ell)^{-\ell}] \cdot f(OPT) + \linearFunction(OPT) - \eps M
    \enspace,
\end{align*}
where the first inequality follows from the fact that $\hat{S}$ is the better solution among $\pi_{[\ell]}(S)$ and $\pi_{[\ell]}(S_0)$ (notice that $S_0$ is a deterministic set), the second inequality follows from Corollary~\ref{cor:sum}, and the penultimate inequality holds since $f$ is non-negative and $\linearFunction(S_0) \geq \linearFunction(OPT)$ by the definition of $S_0$.

We now recall that the proof of Proposition~\ref{prop:approximation_basic} shows that when $\ell$ is set to $1 + \lceil 1 / \eps \rceil$ it holds that $(1 + 1/\ell)^{-\ell} = 1/e + O(\eps)$. Plugging this observation into the previous inequality yields
\begin{align*}
    \bE[f(\hat{S}) + \linearFunction(\hat{S})]
    \geq{} &
    (1 - 1/e - O(\eps)) \cdot f(OPT) + \linearFunction(OPT) - \eps M\\
    \geq{} &
    (1 - 1/e) \cdot f(OPT) + \linearFunction(OPT) - O(\eps) \cdot M
    \enspace,
\end{align*}
where the second inequality holds since $f(OPT) \leq M$ by the definition of $M$. The above inequality matches the guarantee of Theorem~\ref{thm:sum} for the value of the output set $\hat{S}$, and thus, completes the proof of this theorem.
\end{proof}

\paragraph*{Acknowledgments.}

We would like to thank Wenxin Li for suggesting some of the modifications in the current version of Inequality~\eqref{ineq-localOPT1} and the application of our technique to the problem of maximizing sums of monotone submodular functions and linear functions. We would like to thank anonymous reviewers for improving the readability of the current version, and for pointing out some correctness issues in our original proof of Theorem~\ref{thm:sum}.  
The work of Niv Buchbinder was supported in part by Israel Science Foundation (ISF) grant no. 3001/24 and United States - Israel Binational Science Foundation (BSF) grant no. 2022418. The work of Moran Feldman was supported in part by Israel Science Foundation (ISF) grant no. 459/20.

\bibliographystyle{plain}
\bibliography{submodular}

\end{document}